\documentclass[letterpaper]{article} 
\usepackage{aaai24}  
\usepackage{times}  
\usepackage{helvet}  
\usepackage{courier}  
\usepackage[hyphens]{url}  
\usepackage{graphicx} 
\usepackage{natbib}  
\usepackage{caption} 
\title{Learning against Non-credible Auctions}
\author{
    Qian Wang\textsuperscript{\rm 1},
    Xuanzhi Xia\textsuperscript{\rm 2},
    Zongjun Yang\textsuperscript{\rm 3},
    Xiaotie Deng\textsuperscript{\rm 1},
    Yuqing Kong\textsuperscript{\rm 1},\\
    Zhilin Zhang\textsuperscript{\rm 4},
    Liang Wang\textsuperscript{\rm 4},
    Chuan Yu\textsuperscript{\rm 4},
    Jian Xu\textsuperscript{\rm 4},
    Bo Zheng\textsuperscript{\rm 4},
}
\affiliations{
    \textsuperscript{\rm 1} Center on Frontiers of Computing Studies, School of Computer Science, Peking University\\
    \textsuperscript{\rm 2} Yuanpei College, Peking University\\
    \textsuperscript{\rm 3} School of Electronics Engineering and Computer Science, Peking University\\
    \textsuperscript{\rm 4} TAOBAO \& TMALL Group
}

\usepackage{bibentry}

\usepackage{url}            
\usepackage{booktabs}       
\usepackage{amsfonts}       
\usepackage{nicefrac}       
\usepackage{microtype}      
\usepackage{xcolor}         

\usepackage{amsmath}
\usepackage{amssymb}
\usepackage{mathtools}
\usepackage{amsthm}
\usepackage{thm-restate}
\usepackage{bm}
\usepackage{multirow}
\usepackage[inline]{enumitem}
\usepackage[ruled,linesnumbered,algosection]{algorithm2e}

\usepackage[capitalize,noabbrev]{cleveref}
\theoremstyle{plain}
\newtheorem{theorem}{Theorem}[section]
\newtheorem{proposition}[theorem]{Proposition}
\newtheorem{lemma}[theorem]{Lemma}
\newtheorem{corollary}[theorem]{Corollary}

\theoremstyle{definition}

\newtheorem{assumption}[theorem]{Assumption}
\theoremstyle{remark}

\newtheorem{example}[theorem]{Example}

\crefname{assumption}{Assumption}{Assumptions}

\usepackage{xcolor}
\definecolor{pku-red}{RGB}{139,0,18}

\newcommand{\calB}{\mathcal{B}}

\newcommand{\calH}{\mathcal{H}}

\newcommand{\calL}{\mathcal{L}}

\newcommand{\calR}{\mathcal{R}}

\newcommand{\calV}{\mathcal{V}}

\newcommand{\bbE}{\mathbb{E}}

\newcommand{\bbG}{\mathbb{G}}

\newcommand{\bbR}{\mathbb{R}}

\newcommand{\vd}{\boldsymbol{d}}

\newcommand{\vv}{\boldsymbol{v}}

\newcommand{\balpha}{\boldsymbol{\alpha}}

\begin{document}

\maketitle

\begin{abstract}

The standard framework of online bidding algorithm design assumes that the seller commits himself to faithfully implementing the rules of the adopted auction. However, the seller may attempt to cheat in execution to increase his revenue if the auction belongs to the class of non-credible auctions. For example, in a second-price auction, the seller could create a fake bid between the highest bid and the second highest bid. This paper focuses on one such case of online bidding in repeated second-price auctions. At each time $t$, the winner with bid $b_t$ is charged not the highest competing bid $d_t$ but a manipulated price $p_t = \alpha_0 d_t + (1-\alpha_0) b_t$, where the parameter $\alpha_0 \in [0, 1]$ in essence measures the seller's credibility. Unlike classic repeated-auction settings where the bidder has access to samples $(d_s)_{s=1}^{t-1}$, she can only receive mixed signals of $(b_s)_{s=1}^{t-1}$, $(d_s)_{s=1}^{t-1}$ and $\alpha_0$ in this problem. The task for the bidder is to learn not only the bid distributions of her competitors but also the seller's credibility. We establish regret lower bounds in various information models and provide corresponding online bidding algorithms that can achieve near-optimal performance. Specifically, we consider three cases of prior information based on whether the credibility $\alpha_0$ and the distribution of the highest competing bids are known. Our goal is to characterize the landscape of online bidding in non-credible auctions and understand the impact of the seller's credibility on online bidding algorithm design under different information structures.

\end{abstract}
\section{Introduction}

Digital advertising has experienced significant expansion due to the rapid rise of online-activities, surpassing traditional advertising as the dominant marketing influence in various industries.
Between 2021 and 2022, digital advertising revenues in U.S. grew 10.8\% year-over-year totalling \$209.7 billion dollars~\cite{iab2022report}. 
In practice, a huge amount of online ads are sold via real-time auctions implemented on advertising platforms and advertisers participate in such repeated online auctions to purchase advertising opportunities.
This has motivated a flourishing line of work to focus on the problem of online bidding algorithm design. 
In particular, learning to bid in repeated second-price auctions---often with constraints or unknown own valuations---has been well studied due to the popularity of this auction format in practice \cite{weed2016online,balseiro2019learning,golrezaei2021bidding,balseiro2022best,feng2022online,chen2022dynamic,chen2023coordinated}. However, these studies are in fact based on an implicit assumption that the seller commits himself to faithfully implementing the rules of the announced second-price auction.

The possibility that a seller can profitably cheat in a second-price auction was pointed out as early as the seminal paper \cite{vickrey1961counterspeculation} that introduced this auction format. 
After observing all the bids, the seller can strategically exaggerate the highest competing bid and overcharge the winner up to the amount of her own bid. 
Several following papers studied the issue of seller cheating in second-price auctions \cite{rothkopf1995two,porter2005cheating,mcadams2007pays}.
Recent theoretical work by \citet{akbarpour2020credible} formally modelled \textit{credibility} in an extensive-form game where the seller is allowed to deviate from the auction rules as long as the deviation cannot be detected by bidders. 
They defined an auction to be \textit{credible} if it is incentive-compatible for the seller to follow the rules in the presence of cheating opportunities. 
In a second-price auction, the seller can even charge the winner the amount of her own bid to obtain higher utility, with an innocent explanation that the highest and second-highest bids are identical. 
Therefore, the prevalent second-price auction belongs to the class of non-credible auctions in this framework. 
Taking credibility into consideration, advertisers are confronted with a question of practical importance: how should an advertiser bid in repeated non-credible second-price auctions to maximize her cumulative utility?

In this work, we formulate the above problem as an online learning problem for a single bidder. We consider the scenario with a single seller who runs repeated non-credible second-price auctions. At each time~$t$, the winner with bid~$b_t$ is charged not the highest competing bid~$d_t$ but a manipulated price $p_t = \alpha_0 d_t + (1-\alpha_0) b_t$ for some $\alpha_0 \in [0, 1]$. The parameter $\alpha_0$ in essence captures the seller's credibility, the extent to which the seller deviates from the second-price auction rules. 
Our linear model is equivalent to the classic bid-shilling model \cite{rothkopf1995two,porter2005cheating,mcadams2007pays} in expectation. In the bid-shilling model, the seller cheats by inserting a shill bid after observing all of the bids with probability $P^c$. The seller is assumed to take full advantage of his power so the winner will pay her own bid if the seller does cheat. Then the winner's expected payment is $P^c b_t + (1-P^c) d_t$. Moreover, no matter what charging rules the seller actually uses, as long as the estimated credibility within this linear model is away from $1$, it can be confirmed that the seller is cheating. We believe our results and techniques have implications for the more complicated setting with $p_t = h(b_t, d_t;\alpha)$.

We assume the highest competing bids $(d_t)_{t=1}^T$ are \textit{i.i.d.} sampled from a distribution $G$. 
The bidder aims to maximize her expected cumulative utility, which is given by the expected difference between the total value and the total payment. 
Moving to the information model, We investigate three cases of prior information: \begin{enumerate*}[label={(\arabic*)}]
    \item known $\alpha_0$ and unknown $G$;
    \item unknown $\alpha_0$ and known $G$;
    \item unknown $\alpha_0$ and unknown $G$.
\end{enumerate*}
For all three cases, we consider bandit feedback where the bidder can observe the realized allocation and cost at each round. 
For the last case where neither is unknown, we additionally consider full feedback where the price $p_t$ is always observable regardless of the auction outcome.
More discussions on modeling will be placed in \Cref{sec: problem formulation}.

The key challenge of this problem lies in the lack of credibility and its impact on the learning process. 
If assuming the seller has full commitment, it is well known that truthful bidding is the dominant strategy in second-price auctions, but this truthful property no longer holds in non-credible second-price auctions.
Identifying optimal bidding strategies for utility maximization requires not only knowing the bidder's own values but also considering the strategies of her competitors and the seller.
In classic repeated-auction settings (assuming a trustworthy seller), online bidding algorithms can collect historical samples $(d_s)_{s=1}^{t-1}$ to estimate distribution $G$. However, the bidder in non-credible auctions needs to cope with an additional dimension of uncertainty: the available observations under either bandit or full feedback are all manipulated prices, i.e., mixed signals of $(b_s)_{s=1}^{t-1}, (d_s)_{s=1}^{t-1}$ and $\alpha_0$.
As a result, difficulties arise in the estimation of the distribution of her competitors' bids and the seller's credibility.

\subsection{Main Contributions}

First, we characterize the optimal clairvoyant bidding strategy in non-credible second-price auctions when the bidder knows both credibility $\alpha_0$ and distribution $G$. This optimal clairvoyant bidding strategy is also used as the benchmark strategy in the regret definition.

Next, we establish regret lower bounds in various information models and provide corresponding online bidding algorithms that are optimal up to log factors. Our results are summarized in \Cref{tab: result summary}.
\begin{itemize}
    \item For the case where $G$ is unknown and $\alpha_0$ is known, we explore the landscape by discussing how the problem varies with different credibility parameter $\alpha_0 = 0, \alpha_0 = 1$ and $\alpha_0 \in(0,1)$. We mainly contribute to the regret analysis for $\alpha_0 \in(0,1)$, with a proven $\Omega(\sqrt{T})$ lower bound, and a concrete near-optimal $\widetilde{O}(\sqrt{T})$ algorithm.
    
    \item For the case where $G$ is known and $\alpha_0$ is unknown, we develop an $O(\log^2{T})$ algorithm, which adopts a dynamic estimation approach to approximate  $\alpha_0$.
    
    \item For the challenging case where both $G$ and $\alpha_0$ are unknown, we observe that under bandit feedback, an $\Omega(T^{2/3})$ lower bound and an $\widetilde{O}(T^{2/3})$ algorithm follow directly from existing algorithms. We then turn to the more interesting setting with full information feedback, for which we propose an episodic bidding algorithm that learns $\alpha_0$ and $G$ simultaneously in an efficient manner, while achieving a near-optimal regret of $\widetilde{O}(\sqrt{T})$.
\end{itemize}

Overall, this work provides a theoretical regret analysis for learning against non-credible auctions. We aim to characterize the landscape of online bidding in non-credible auctions and analyze how the seller' credibility influences the design of online bidding algorithms under different information structures.

\begin{table*}
    \caption{Result Summary.}
    \label{tab: result summary}
    \centering
    \begin{tabular}{cccccc}
        \toprule
        $\balpha_0$ & $\boldsymbol{G}$ & \textbf{Feedback} & \textbf{Upper bound} & \textbf{Lower bound} & \textbf{Theorem} \\
        \midrule
        Known, $\alpha_0 = 1$ & \multirow{3}{*}{Unknown} & \multirow{3}{*}{Bandit} & 0 & 0 & \multirow{3}{*}{\Cref{thm: known alpha unknown G}$\phantom{^2}$} \\
        \cmidrule(r){1-1}\cmidrule(r){4-5}
        Known, $\alpha_0 \in (0, 1)$ & & & $\widetilde{O}(T^{1/2})$ & $\Omega(T^{1/2})$ & \\
        \cmidrule(r){1-1}\cmidrule(r){4-5}
        Known, $\alpha_0 = 0$ & & & $\widetilde{O}(T^{2/3})$ & $\Omega(T^{2/3})$ & \\
        \cmidrule(r){1-6}
        Unknown & Known & Bandit & $\widetilde{O}(1)$ & $\Omega(1)$ & \Cref{thm: unknown alpha known G}$^1$ \\
        \cmidrule(r){1-6}
        \multirow{2}{*}{Unknown} & \multirow{2}{*}{Unknown} & Bandit & $\widetilde{O}(T^{2/3})$ & $\Omega(T^{2/3})$ & \Cref{cor: unknown alpha unknown G} \\
        \cmidrule(r){3-6}
        & & Full & $\widetilde{O}(T^{1/2})$ & $\Omega(T^{1/2})$ & \Cref{thm: unknown alpha unknown G full}$^1$ \\
        \bottomrule
    \end{tabular}
    
    \medskip
    $^1$ The regret bounds of these theorems rely on corresponding assumptions.
\end{table*}

\subsection{Related Work}

\textbf{Credibility in auctions.} The issue of seller cheating has been studied by the game-theoretic literature in a strategic framework ~\citep{mcadams2007pays,porter2005cheating,rothkopf1995two}.
Recently, \citet{akbarpour2020credible} explored the setting where the seller deviates from the auction rules in a way that can be innocently explained. To this end, they defined credibility based on the detectability of seller deviations. 
They further established an impossibility result that no optimal auction simultaneously achieves staticity, credibility and strategy-proofness. They showed that the first-price auction is the unique static optimal auction that achieves credibility. 
\citet{epub94449} considered the general allocation problem and introduced the definition of \textit{verifiability} (i.e., allowing participants to check the correctness of their assignments) and \textit{transparency} (i.e., allowing participants to check whether the allocation rule is deviated). These are stronger security notions than the credibility concept investigated by \cite{akbarpour2020credible}. 
\citet{dianat2021credibility} studied how the credibility of an auction format affects bidding behavior and final outcomes via laboratory experiments. Their empirical findings confirm the theory that sellers do have incentives to break the auction rules and overcharge the winning bidder. These pioneering works discussed how a participant can potentially detect and learn non-credible mechanisms as we do. In contrast, our work is based on the online learning framework, where information revelation is partial and sequential in nature. 

\textbf{Learning to bid.} Our work is closely related with the line of literature on learning to bid in repeated auctions. \citet{balseiro2019contextual,han2020optimal,han2020learning,badanidiyuru2023learning,zhang2022leveraging} studied the problem of learning in repeated first-price auctions. \citet{castiglioni2022online,wang2023learning} studied no-regret learning in repeated first-price auctions with budget constraints. As for repeated second-price auctions, \citet{balseiro2019learning,balseiro2022best,chen2022dynamic} considered the bidding problem with budget constraints and \citet{feng2022online,golrezaei2021bidding} further considered return-on-spend (RoS) constraints.
All these works assume that the seller has full commitment to the announced auction rules.

\section{Problem Formulation}
\label{sec: problem formulation}

We consider the problem of online learning in repeated non-credible second-price auctions. 
We focus on a single bidder in a large population of bidders during a time horizon $T$. 
In each round $t = 1, \ldots, T$, there is an available item auctioned by a single seller. 
The bidder perceives a value $v_t\in [0, 1]$ for this item, and then submits a bid $b_t \in \bbR_+$ based on $v_t$ and all historical observations available to her. 
We denote the maximum bid of all other bidders by $d_t\in \bbR_+$. As usual, we use the bold symbol $\vv$ without subscript $t$ to denote the vector $(v_{1}, \ldots, v_{T})$; the same goes for other variables in the present paper.

We consider a \textit{stochastic} setting where $v_t$ is \textit{i.i.d.} sampled from a distribution $F$ and $d_t$ is \textit{i.i.d.} sampled from a distribution $G$. 
The latter assumption follows from the standard mean-field approximation \cite{iyer2014mean,balseiro2015repeated} and is a common practice in literature. 
The main rationale behind this assumption is that when the number of other bidders is large, on average their valuations and bidding strategies are static over time. 
Whether $G$ is known to the bidder depends on the information structure while $F$ is always unknown to the bidder.

The seller claims that he follows the rules of the second-price auction, but he actually uses a combination of the first-price auction and the second-price auction. 
For simplicity, we assume a linear model. 
The auction outcome in round $t$ is then as follows: if $b_t \geq d_t$, the bidder wins the item and pays $p_t \coloneqq \alpha_0 d_t + (1 - \alpha_0)b_t$, where $\alpha_0$ is assumed to be a fixed weight throughout the period; if $b_t < d_t$, the bidder loses the auction and pays nothing. 
Here we assume that ties are broken in favor of the bidder we concern to simplify exposition.
We only consider $\alpha_0 \in [0, 1]$ since $\alpha_0 < 0$ will be immediately detected by the winner and $\alpha_0 > 1$ will lead to lower revenue than mere second-price auctions. 
One can observe the pricing rule follows the second-price auction when $\alpha_0 = 1$, and it follows the first-price auction when $\alpha_0 = 0$. 

Let $x_{t} \coloneqq \mathbb{I}\left\{b_{t} \geq d_{t}\right\}$ be the binary variable indicating whether the bidder wins the item. 
Let $c_{t} \coloneqq x_{t}p_{t}$ be the bidder's cost and let $r_{t} \coloneqq x_{t}v_{t} - c_{t}$ be the corresponding reward. 

\paragraph{Information structure.} In this paper, we investigate three cases of prior information: \begin{enumerate*}[label={\arabic*)}]
    \item known credibility $\alpha_0$ and unknown distribution $G$;
    \item unknown credibility $\alpha_0$ and known distribution $G$;
    \item unknown credibility $\alpha_0$ and unknown distribution $G$.
\end{enumerate*}

The first two cases not only serve as warm-up analysis, but also have practical significance in their own right. In reality, a bidder may receive some additional signals beyond the learning process to construct her belief over the seller's credibility or the strategies of other bidders, e.g. the seller's reputation heard from other bidders, bidding data collected through other credible channels. Then her bidding algorithm mainly aims to learn the other part in the competing environment.

We consider two cases of information feedback:
\begin{enumerate}
    \item \textit{Bandit information feedback. } The bidder can observe the allocation $x_t$ and the cost $c_t$ at the end of each round $t$. 
    \item \textit{Full information feedback.} The bidder can observe the allocation $x_t$ and the price $p_t$ at the end of each round $t$.
\end{enumerate}
The second information feedback makes sense in non-censored auctions where the seller-side platform (SSP) is supposed to provide the minimum winning price for every bidder regardless of the outcome. If the bidder wins, a dishonest seller will overstate the minimum winning price to overcharge the winner; if the bidder loses, a dishonest seller will understate the minimum winning price to trick the bidder into raising her bids in the following rounds. The full-feedback model for simplicity assumes that these two types of deceptions are symmetric, controlled by the same parameter $\alpha_0$. Note that when $\alpha_0 = 0$, both feedback models are equivalent to the binary feedback model in first-prices auctions.

We denote the historical observations available to the bidder before submitting a bid in round $t$ by $\calH_t$. For the two cases of information feedback, we have, respectively, 
\begin{align*}
    \calH^B_t \coloneqq \left(v_{s}, x_{s}, c_{s}\right)_{s=1}^{t-1}, \ \calH^F_t \coloneqq \left(v_{s}, x_{s}, p_{s}\right)_{s=1}^{t-1}.
\end{align*}
We will omit the superscript $B$ or $F$ in the remaining of this paper when the context is clear.

\paragraph{Bidding strategy and regret.} A bidding strategy maps $(\calH_t, v_t)$ to a (possibly random) bid $b_t$ for each $t$. For a strategy $\pi$, we denote by $\calR(\pi)$ the expected performance of $\pi$, defined as follows:
\begin{align*}
    \calR(\pi) & =  \bbE^{\pi}_{\vv, \vd} \left[\sum_{t=1}^T r^{\pi}_t\right] \\ 
    & = \bbE^{\pi}_{\vv, \vd} \left[\sum_{t=1}^T \mathbb{I} \left\{b^{\pi}_{t} \geq d_{t}\right\}\left(v_t - p_t^{\pi}\right)\right],
\end{align*}
where the expectation is taken with respect to the values $\vv$, the highest competing bids $\vd$ and any possible randomness embedded in the strategy $\pi$. The expect regret of the bidder is defined to be the difference in the expected cumulative rewards of the bidder's strategy and the optimal bidding strategy, which has the perfect knowledge of $\alpha_0$, $F$ and $G$ to maximize the target:
\begin{align*}
    \mathrm{Regret}(\pi) = \max_{\pi'} \calR(\pi') - \calR(\pi).
\end{align*}
Now we look at the the optimal bidding strategy. Using the independence of $d_t$ from $v_t$, one has
\begin{align*}
    \calR(\pi) = \sum_{t=1}^T \bbE^{\pi}_{\calH_t, v_t} \left[(v_t - b_t^{\pi}) G(b_t^{\pi}) + \alpha_0\int_{0}^{b_t^{\pi}}G(y)dy \right].
\end{align*}
Let $r(v, b, \alpha) = \left(v - b\right)G(b) + \alpha\int_{0}^{b}G(y)\,dy$ and let $b^*(v, \alpha) = \arg\max_{b} r(v, b, \alpha)$ (taking the largest in a case of a tie). Then with the perfect knowledge of $\alpha_0$ and $G$, the optimal strategy in each round submits $b^*(v_t, \alpha_0)$, denoted by $b_t^*$ for short. The expect regret of strategy $\pi$ can written as
\begin{align*}
    \text{Regret}(\pi) = \bbE^{\pi}_{\vv, \vd} \left[\sum_{t=1}^T r(v_t, b_t^*, \alpha_0) - \sum_{t=1}^T r(v_t, b_t^{\pi}, \alpha_0)\right].
\end{align*}
We will omit the superscript $\pi$ in the remaining of this paper when the context is clear .

\section{Learning with Known $\alpha$ and Unknown $G$}
\label{sec: known alpha unknown G}

We start with the scenario where the credibility parameter $\alpha$ is known but the distribution $G$ of the highest competing bids is unknown. The algorithms and proofs of this section are deferred to \Cref{app: known alpha unknown G}.

Observe that when $\alpha$ reaches an endpoint of $[0, 1]$, this problem will degenerate into a bidding problem in repeated first-price or second-price auctions. Therefore, We will consider three cases separately: $\alpha_0 = 0$, $\alpha_0 = 1$, and $\alpha_0 \in (0,1)$. The following \Cref{thm: known alpha unknown G} provides a comprehensive characterization of this setting. For each case, it establishes a regret lower bound and gives an algorithm with optimal performance up to log factors. 

\begin{restatable}{theorem}{firsttheorem}
\label{thm: known alpha unknown G}
For repeated non-credible second-price auctions with known credibility $\alpha_0$, unknown distribution $G$ and bandit feedback:
\begin{enumerate}[label={(\arabic*)}]
    \item when $\alpha_0 = 1$, truthful bidding achieves no regret;
    \item when $\alpha_0 = 0$, there exists a bidding algorithm (\Cref{algo: known alpha unknown G fpa}) that achieves an $\widetilde{O}(T^{2/3})$ regret, and the lower bound on regret for this case is $\Omega(T^{2/3})$;
    \item when $\alpha_0 \in (0,1)$, there exists a bidding algorithm (\Cref{algo: known alpha unknown G main}) that achieves an $\widetilde{O}(T^{1/2})$ regret, and the lower bound on  regret for this case is $\Omega(T^{1/2})$.
\end{enumerate}
\end{restatable}

The lower bounds in \Cref{thm: known alpha unknown G} show that the hardness of this learning problem increases as $\alpha_0$ decreases, which demonstrates the impact of reduced credibility on online bidding optimization. When $\alpha_0$ deviates from $1$, truthful bidding is no longer a dominant strategy. The bidder has to learn the competitors' bids to make decisions, and thus the distribution estimation error would introduce an inevitable regret of order $\Omega(T^{1/2})$. As long as $\alpha_0 > 0$, the bidder can infer the highest competing bid $d_t$ from her payment once she wins an auction by measuring the difference between $c_t$ and $b_t$. However, when $\alpha_0$ becomes $0$, $c_t \equiv b_t$ in winning rounds provides no additional information about $d_t$. The complete loss of credibility cripples the bidder's ability to observe the competitive environment and estimate the distribution $G$, so the regret lower bound leaps from $\Omega(T^{1/2})$ to $\Omega(T^{2/3})$.

The first statement of \Cref{thm: known alpha unknown G} is trivial due to the nature of the second-price auction. 

The second case is equivalent to bidding in repeated first-price auctions with binary feedback (receiving only $x_t = 1$ or $0$), which can be modeled as a contextual bandits problem with cross learning:
\begin{itemize}
    \item \textit{Cross learning over contexts. } The bidder in round $t$ not only receives the reward $r_t$ under $(v_t, b_t)$, but also observes the rewards $r'$ under $(v', b_t)$ for every other $v'$. 
\end{itemize}
For the contextual bandits problem with cross-learning over contexts in the stochastic setting, \citet{balseiro2019contextual} proposed a UCB-based algorithm that can achieve an $O(\sqrt{KT})$ regret, where $K$ is the number of actions. Applying this algorithm to the auction setting results in a regret bound of $\widetilde{O}(\sqrt{KT}+T/K)$, where the last term comes from the discretization error and the upper bound becomes $\widetilde{O}(T^{2/3})$ with $K\sim T^{1/3}$. They also proved the regret lower bound is $\Omega(T^{2/3})$ via a reduction to the problem of dynamic pricing.

\begin{lemma}[\citet{balseiro2019contextual}]
\label{lem: known alpha unknown G fpa}
For repeated first-price auctions with binary feedback, \Cref{algo: known alpha unknown G fpa} can achieve an $\widetilde{O}(T^{2/3})$ regret, and there exists a problem instance where any algorithm must incur a regret of at least $\Omega(T^{2/3})$ regret.
\end{lemma}

The third case is similar to bidding in repeated first-price auctions with censored feedback, where the seller runs first-price auctions and always reveals the winner's bid to each bidder so the bidder can see the highest competing bid only if she loses the auction. \citet{han2020optimal} modeled that problem as a contextual bandits problem with cross learning, partial ordering:
\begin{itemize}
    \item \textit{Cross learning over contexts.}
    \item \textit{Partial ordering over actions. } There exists a partial order $\preceq_\calB$ over the action set $\calB$. The bidder in round $t$ not only receives the reward $r_t$ under $(v_t, b_t)$, but also observes the rewards $r'$ under $(v_t, b')$ for every other $b' \preceq_\calB b_t$. 
    \item \textit{Partial ordering over contexts. } 
    There exists a partial order $\preceq_\calV$ over the context set $\calV$ such that if $v_1 \preceq_\calV v_2$, then $b^*(v_1) \preceq_\calB b^*(v_2)$ where $b^*(v)$ is the optimal auction under context $v$.
\end{itemize}
\begin{restatable}{lemma}{knownalphaupper}[\citet{han2020optimal}]
\label{lem: known alpha unknown G main upper}
For the contextual bandits problem with cross-learning over contexts, partial ordering over auctions and contexts in the stochastic setting, there exists a bidding algorithm than achieves an $\widetilde{O}(T^{1/2})$ regret.
\end{restatable}

We carefully adapt their algorithm to our setting (\Cref{algo: known alpha unknown G main}) and verify that the third case with $\alpha_0 \in (0, 1)$ satisfies the above three properties:
\begin{itemize}
    \item \textit{Cross learning over values.} 
    The bidder can calculate $r'$ under $(v', b_t)$ by
    \begin{align}
    \label{eqn: infer r_t'}
        r' = r_t + x_t(v' - v_t).
    \end{align}
    \item \textit{Partial ordering over bids. } If the bidder wins in round $t$, she can infer the highest competing bid $d_t$ by
    \begin{align}
    \label{eqn: infer d_t}
        d_t = \left(c_t - (1-\alpha_0)b_t\right)/\alpha_0.
    \end{align}
    Therefore the reward $r'$ under $(v_t, b')$ for any other $b < b_t$ can be calculated by using the corresponding allocation $\mathbb{I}\left\{b \geq d_{t}\right\}$ and price $\alpha_0 d_t + (1-\alpha_0) b$. If the bidder loses the auction with $b_t$, she should also lose with $b' < b_t$ and the reward $r'$ is $0$.
    \item \textit{Partial ordering over values. } We have shown in the previous section that the optimal bid under value $v$ is $b^*(v, \alpha_0) = \arg\max_{b} r(v, b, \alpha_0)$. The following lemma shows it is a non-decreasing function in $v$.
\end{itemize}

\begin{restatable}{lemma}{increasing}
\label{lem: optimal bid monotonicity}
$b^*(v, \alpha) = \arg\max_{b} r(v, b, \alpha)$ (taking the largest in the case of a tie) is a non-decreasing function in both $v$ and $\alpha$.
\end{restatable}

Therefore, \Cref{algo: known alpha unknown G main} can achieve an $\widetilde{O}(T^{1/2})$ regret when $\alpha_0 \in (0, 1)$. For the last piece of the puzzle, we prove the following regret lower bound. Remark that \Cref{lem: known alpha unknown G main lower} holds for any $\alpha_0$, though the bound is not tight when $\alpha_0 = 0$.

\begin{restatable}{lemma}{knownalphalower}
\label{lem: known alpha unknown G main lower}
For repeated non-credible second-price auctions with known credibility $\alpha_0$, unknown distribution $G$, there exists a constant $c > 0$ such that 
$$
\inf_{\pi} \sup_{G}\mathrm{Regret}(\pi) \geq c \cdot (1-\alpha_0) \sqrt{T}, 
$$
even in the special case with $v_t \equiv 1$ and full feedback.
\end{restatable}

\section{Learning with Unknown $\alpha$ and Known $G$}
\label{sec: unknown alpha known G}

We next consider the scenario where distribution $G$ is known but credibility $\alpha_0$ is unknown. The proofs of this section are deferred to \Cref{app: unknown alpha known G}.

\begin{algorithm}[ht]
\caption{Learning with unknown $\alpha_0$, known $G$ and bandit feedback}
\label{algo: unknown alpha known G}
    \SetKwInput{KwInit}{Initialization}
    \KwIn{Time horizon $T$; distribution $G$.}
    \KwInit{The bidder submits $b_1 = 1$.}
    \For{$t \gets 2$ \KwTo $T$}{
        The bidder receives the value $v_t \in [0, 1]$. \\
        The bidder estimates the seller's credibility by
        \begin{align}
        \label{eqn: estimate alpha unknown_alpha_known_G}
            \widetilde{\alpha}_t = \arg\min_{\alpha\in [0, 1]}\left|\sum_{s=1}^{t-1} \left(r_s - r(v_s, b_s, \alpha)\right)\right|,
        \end{align}\\
        The bidder submits $b_t = \arg\max_{b} r(v_t, b, \widetilde{\alpha}_t)$.
    }
\end{algorithm}

Our bidding algorithm for this setting is depicted in \Cref{algo: unknown alpha known G}. 
The bidder first conducts a one-round exploration to make an appropriate initialization. 
After receiving the value $v_t$ in each round $t = 2, \ldots, T$, the bidder computes $\widetilde{\alpha}_t$, which is the estimation of $\alpha_0$ based on the historical observations in the past $t-1$ rounds. 
Recall that the optimal bid $b^*_t$ shown in \Cref{sec: problem formulation} maximizes $r(v_t, b, \alpha_0)$. 
Thus, by the choice of $b_t$, if the estimator $\widetilde{\alpha}_t$ is close to $\alpha_0$, the expected reward $r(v_t, b_t, \alpha_0)$ will be close to the optimal reward $r(v_t, b_t^*, \alpha_0)$.

\Cref{eqn: estimate alpha unknown_alpha_known_G} in \Cref{algo: unknown alpha known G} aims to estimate $\alpha_0$ by fitting the observed rewards $\{r_s\}_{s=1}^{t-1}$ to the expected rewards $\{r(v_s, b_s, \alpha)\}_{s=1}^{t-1}$, which can be computed given $G$ is known. We apply the Azuma-Hoeffding inequality to obtain the following result.

\begin{restatable}{lemma}{alphaconvergence}
\label{lem: convergence of alpha unknown_alpha_known_G}
Under \Cref{algo: unknown alpha known G}, we have with probability at least $1-\delta$, $\forall t\in [T]$, 
\begin{align*}
    |\widetilde{\alpha_t} - \alpha_0| \leq w_t,
\end{align*}
where $w_t$ is given by
\begin{align*}
    w_t = \frac{2\sqrt{2(t-1)\log{(2T/\delta)}}}{\sum_{s=1}^{t-1} \int_{0}^{b_s}G(y)dy}.
\end{align*}
\end{restatable}

For technical purpose, we make the following assumption.

\begin{assumption}
\label{asm: unknown alpha known G}
    $G$ is twice differentiable and log-concave with density function $g$. There exist positive constants $B_1, B_2$ such that $B_1 \leq g(x) \leq B_2$ for $x\in [0, 1]$.
\end{assumption}

The CDFs of many common distributions, such as gamma distributions, Gaussian distributions and uniform distributions, are all log-concave. The existence of positive bounds on the density function is also a standard and common assumption in various learning problems.

\begin{restatable}{theorem}{secondtheorem}
\label{thm: unknown alpha known G}
Suppose that \Cref{asm: unknown alpha known G} holds. For repeated non-credible second-price auctions with unknown credibility $\alpha_0$, known distribution $G$ and bandit feedback, there exists a bidding algorithm (\Cref{algo: unknown alpha known G}) that achieves an $O(\log^2{T})$ regret, and any algorithm must incur at least a constant regret.
\end{restatable}

Remark that due to the additional assumption, we cannot actually draw the conclusion that estimating $\alpha$ is generally easier than estimating $G$ by comparing two lower bounds in \Cref{thm: unknown alpha known G} and \Cref{thm: known alpha unknown G}. For example, the two-point method used in the proof of \Cref{lem: known alpha unknown G main lower} constructs two discrete $G$ distributions to show no bidding algorithm can obtain low regret simultaneously under both distributions, while \Cref{asm: unknown alpha known G} has ruled out such bad cases.

A key step in the proof of \Cref{thm: unknown alpha known G} involves showing that the reward per round obtained by the bidder is close to the reward under optimal bid with high probability. We have with probability at least $1-\delta$, $\forall t$, 
\begin{align*}
    r(v_t, b_t^*, \alpha_0) - r(v_t, b_t, \alpha_0) \leq w_t^2/B_1.
\end{align*}
Then the regret upper bound can be established by showing $w_t \sim \sqrt{\log T/t}$ and summing up through $1$ to $T$. It is also worth discussing the robustness of this result. Even without \Cref{asm: unknown alpha known G}, we can get
\begin{align*}
    r(v_t, b_t^*, \alpha_0) - r(v_t, b_t, \alpha_0) \leq w_t.
\end{align*}
And with a weaker continuity condition, it still holds that $w_t \sim \sqrt{\log T/t}$.

\begin{restatable}{corollary}{secondcorollary}
\label{cor: unknown alpha known G}
Suppose that $G$ is continuous. For repeated non-credible second-price auctions with unknown credibility $\alpha_0$, known distribution $G$ and bandit feedback, \Cref{algo: unknown alpha known G} can achieve an $\widetilde{O}(\sqrt{T})$ regret.
\end{restatable}

An intuition on why the proof of the regret upper bound may fail under some discontinuous distributions is given in \Cref{eg: bad example unknown_alpha_known_G}. In spite of this, we conjecture that \Cref{algo: unknown alpha known G} can also guarantee a lower regret in the case with discontinuous $G$ and we leave that as a future direction.
\section{Learning with Unknown $\alpha$ and Unknown $G$}
\label{sec: unknown alpha unknown G}

\paragraph{Bandit feedback.} For the last scenario with both $\alpha$ and $G$ unknown, we first consider bandit feedback. Although the seller's credibility is unknown, this case still satisfies the cross-learning property over values, i.e., \Cref{eqn: infer r_t'} holds. Thus, \Cref{algo: known alpha unknown G fpa}, which actually does not use the value of $\alpha_0$, can still work in this case and achieve an $\widetilde{O}(T^{2/3})$ regret. The regret lower bound also directly follows the third statement of \Cref{thm: known alpha unknown G} since any bidding strategy cannot obtain a better regret guarantee than $O(T^{2/3})$ when $\alpha_0 = 0$. Hence, we have the following result. 

\begin{corollary}
\label{cor: unknown alpha unknown G}
For repeated non-credible second-price auctions with unknown credibility $\alpha_0$, unknown distribution $G$ and bandit feedback, there exists a bidding algorithm (\Cref{algo: known alpha unknown G fpa}) that achieves an $\widetilde{O}(T^{2/3})$ regret, and the lower bound on regret for this problem is $\Omega(T^{2/3})$.
\end{corollary}

In spite of the sublinear regret, the result of \Cref{cor: unknown alpha unknown G} is not satisfactory. \Cref{algo: known alpha unknown G fpa} only falls into the category of usual UCB policies with cross learning, but does not make full use of the properties of this auction setting. In fact, it treats the seller's mechanism as a black box without really estimating the seller's credibility or other bidder's strategies.

\paragraph{Full feedback.} With the above in mind, we explore whether we can make any improvements with richer feedback and more meticulous estimation of $\alpha_0$ and $G$. In the full feedback model, $p_t = \alpha_0 d_t + (1-\alpha_0)b_t$ is always observable, which would intuitively help our estimation. Nevertheless, when $\alpha_0 = 0$, binary feedback in first-prices auctions still results in the $\Omega(T^{2/3})$ lower bound. Thus, in what follows we exclude the extreme case of $\alpha=0$ when it is impossible to achieves a better bound than $O(T^{2/3})$. Note that sellers in reality often do not set $\alpha_0 = 0$ due to concerns about reputation or cheating costs \cite{mcadams2007pays}. This assumption is also consistent with the recent empirical findings by \citet{dianat2021credibility} that although sellers in non-credible second-price auctions often overcharge, they typically do not use the rules of the first-price auction to maximize revenue.

\begin{assumption}
\label{asm: unknown alpha unknown G}
There exists a constant $\underline{\alpha} > 0$ such that $\alpha_0 \in [\underline{\alpha}, 1]$.
\end{assumption}

We also make a slightly stronger assumption over $G$.
\begin{assumption}
\label{asm: unknown alpha unknown G 2}
    $G$ is twice differentiable with log-concave density function $g$. There exist positive constants $B_1, B_2, B_3, W$ such that $B_1 < g(x) < B_2$ and $|g'(x)| \leq B_3$ for $x\in [0, 1 + W]$.
\end{assumption}

\begin{algorithm}
\caption{Learning with unknown $\alpha_0$, unknown $G$ and full feedback}
\label{algo: unknown alpha unknown G}
    \SetKwInput{KwInit}{Initialization}
    \KwIn{Time horizon $T$.}
    \For{$t\in \Gamma_1$}{
    The bidder receives the value $v_t \in [0, 1]$. \\
    The bidder submits a bid $b_t = 1$.
    }
    Estimate $\alpha_0$ by using $\widetilde{\alpha}_1$, which is computed by,
    \begin{align}
        \widetilde{\alpha}_1 = \arg\min_{\alpha \in [\underline{\alpha}, 1]} \calL_1(\alpha),
    \end{align}
    where $\calL_1(\alpha)$ is defined in \Cref{eqn: MLE alpha}.

    The bidder computes $\hat{\Psi}(\cdot; \widetilde{\alpha}_1)$ by
    \begin{align}
        \hat{\Psi}_1(\cdot,\widetilde{\alpha}_1) \quad & \nonumber \\
        = \mathop{\arg\max}\limits_{\Psi \, is \, concave} & \frac{1}{T_s} \sum_{t \in \Gamma_s}\Psi\left(\frac{p_t - (1 - \widetilde{\alpha}_1)b_t}{\widetilde{\alpha}_1}\right) \nonumber \\
        - & \int \exp(\Psi(y;\widetilde{\alpha}_1))dy.
    \end{align}\\
    The bidder estimates $G$ by
    \begin{align}
        \hat{G}_1(\cdot; \widetilde{\alpha}_1) = \int \exp(\hat{\Psi}_1(y; \widetilde{\alpha}_1))dy.
    \end{align}
    \For{$s=2,3,\cdots, S$}{
    \For{$t\in \Gamma_s$}{
        The bidder receives the value $v_t \in [0,1]$.\\
        The bidder submits a bid $b_t$, computed by,
        \begin{align}
            b_t = \arg\max_b \Big[ & \left(v_t - b\right)\hat{G}_{s-1}(b) \nonumber \\
            + \widetilde{\alpha}_{s-1} & \int_{0}^{b}\hat{G}_{s-1}(y)\,dy \Big],
        \end{align}\\
        }
        The bidder updates the estimator for $\alpha_0$ in episode $s$ by using $\widetilde{\alpha}_s$, which is computed by,
        \begin{align}
            \widetilde{\alpha}_s = \arg\min_{\alpha \in [\underline{\alpha}, 1]} \calL_s(\alpha).
        \end{align}
        where $\calL_s(\alpha)$ is defined in \Cref{eqn: MLE alpha}.
        The bidder updates $\hat{\Psi}(\cdot; \widetilde{\alpha}_1)$ by
        \begin{align}
            \hat{\Psi}_s(\cdot,\widetilde{\alpha}_s) \quad & \nonumber \\
            = \mathop{\arg\max}\limits_{\Psi \, is \, concave} & \frac{1}{T_s} \sum_{t \in \Gamma_s}\Psi\left(\frac{p_t - (1 - \widetilde{\alpha}_s)b_t}{\widetilde{\alpha}_s}\right) \nonumber \\
            - & \int \exp(\Psi(y;\widetilde{\alpha}_s))dy.
        \end{align}\\
        The bidder updates the estimation of $G$ by 
        \begin{align}
            \hat{G}_s(\cdot; \widetilde{\alpha}_s) = \int \exp(\hat{\Psi}_1(y; \widetilde{\alpha}_s))dy.
        \end{align}

    }
\end{algorithm}

Our bidding algorithm for full feedback is presented in \Cref{algo: unknown alpha unknown G}. It runs in an episodic manner, similar to \citet{cesa2014regret},\citet{javanmard2019dynamic}, \citet{badanidiyuru2023learning}. During a time horizon $T$, the bidding algorithm is divided into $S$ episodes, each containing $T_s = T^{1-2^{-s}}$ time steps.
Denote $\Gamma_s$ be the time steps in stage $s$, s.t. $|\Gamma_s| = T_s$. For any time step $t$ in the first episode, we set $b_t = 1$ for a proper initialization. For any time step $t$ in episode $s$($s\geq 2$), we use the estimated parameter $\widetilde{\alpha}_{s-1}$ and distribution $\hat{G}_{s-1}$ in the $(s-1)$-th episode to set the bid
\begin{align*}
    b_t = \arg\max_b \left(v_t - b\right)\hat{G}_{s-1}(b) + \widetilde{\alpha}_{s-1}\int_{0}^{b}\hat{G}_{s-1}(y)\,dy,
\end{align*}
and only update these estimators at the end of episode $s$ by using the data observed in episode $s$. 

The main difficulty is how to update the estimators of $G$ and $\alpha_0$ in each episode. Recall that when G is known, the estimation step (\Cref{eqn: estimate alpha unknown_alpha_known_G}) in \Cref{algo: unknown alpha known G} is essentially similar to using the maximum likelihood estimation (MLE) method to find the most probable $\alpha$. However, without the knowledge of $G$, the bidder cannot directly estimate $\alpha_0$ by matching observed rewards to expected rewards. To handle this challenge, we combine the non-parametric log-concave density estimator and MLE method, to learn $\alpha_0$ and $G$ simultaneously.

We first introduce the non-parametric estimator of log-concave density function $g$, which is adopted from \citet{dumbgen2009maximum}. In each episode s, given
realized $p_t, t \in \Gamma_s$, \Cref{algo: unknown alpha unknown G} for any parameter $\alpha$ gives an estimator $\hat{g}_s(\cdot,\alpha)$, 
\begin{align}
    \hat{g}_s(\cdot,\alpha) = \mathop{\arg\max}\limits_{\text{$g$ is log-concave}} \frac{1}{T_s} & \sum_{t \in \Gamma_s}\log g\left(\frac{p_t - (1 - \alpha)b_t}{\alpha}\right) \nonumber \\ 
    & - \int g(y)dy.
\end{align}
In this work, we restrict the function class of $\hat{g}_s(\cdot; \alpha)$ for any $\alpha$ and $s$ as below,
\begin{align*}
    \mathcal{P} = \big\{p: p(z) \leq B_2, \forall z \in [0, 1 + W], \int p(z)dz = 1\big\}.
\end{align*}

It is \textit{w.l.o.g.} to re-parameterize $g(y) = \exp(\Psi(y))$, where $\Psi(y)$ is a concave function \textit{w.r.t} $y$. Then it is equivalent to get an estimator $\hat{\Psi}_s(\cdot,\alpha)$, 
\begin{align}
    \hat{\Psi}_s(\cdot,\alpha) = \mathop{\arg\max}\limits_{\text{$\Psi$ is concave}} \frac{1}{T_s} & \sum_{t \in \Gamma_s}\Psi\left(\frac{p_t - (1 - \alpha)b_t}{\alpha}\right) \nonumber \\
    & - \int \exp(\Psi(y;\alpha))dy.
\end{align}

Let $\hat{G}_s(y;\alpha) = \int_{0}^{y} \hat{g}_s(z,\alpha)dz = \int_{0}^{y} \exp(\hat{\Psi}_s(z;\alpha))dz$ be the estimated empirical distribution using the above estimator. Let $\bbG_s$ be the empirical distribution of samples $\{d_t\}_{t \in \Gamma_s}$ in episode $s$ such that $\mathbb{G}_s(y) = \frac{1}{T_s}\sum_{t \in \Gamma_s}\mathbb{I}\{d_t \leq y\}$. \citet{dumbgen2009maximum} proved the following result.

\begin{lemma}[\citet{dumbgen2009maximum}]
\label{lem: C.1}
The optimizer $\hat{\Psi}_s(\cdot,\alpha_0)$ exists and is unique. For any $d_t = (p_t - (1 - \alpha_0)b_t)/\alpha_0, t \in \Gamma_s$, $\mathbb{G}_s(d_t) - \frac{1}{T_s} \leq \hat{G}_s(d_t;\alpha_0) \leq \mathbb{G}_s(d_t)$.
\end{lemma}

Give the above characterization of $\hat{\Psi}_s(\cdot,\alpha_0)$ and $\hat{G}_s(\cdot;\alpha_0)$ we provide the uniform convergence
bound for $|\hat{G}_s(d;\alpha_0) - G(z)|$ in the following lemma.

\begin{restatable}{lemma}{convergencebound}
\label{lem: C.2}
Suppose that $T_s \gg \log^2(2/\delta)$ for some $\delta > 0$, then with probability at least $1 - \delta$, $\forall d \in [0, 1], |\hat{G}_s(d;\alpha_0) - G(z)| \leq O(\sqrt{\log(1/\delta)/T_s})$ holds .
\end{restatable}

Given the non-parametric estimator $\hat{G}_s(\cdot;\alpha)$ introduced in the above, we minimize the following MLE loss function to compute $\widetilde{\alpha}_s$,
\begin{align}
   \calL_s(\alpha) = - \frac{1}{T_s} & \sum_{t \in \Gamma_s}\Big[ \xi_t \log \hat{G}_s(\varepsilon_t(\alpha); \alpha) \nonumber \\
    & + (1- \xi_t) \log(1 - \hat{G}_s(\varepsilon_t(\alpha); \alpha))\Big], \label{eqn: MLE alpha}
\end{align}
where $\xi_t = \mathbb{I}\{p_t \leq b_t + \underline{\alpha}W/2\}$ and $\varepsilon_t(\alpha) = b_t + \underline{\alpha}W/(2\alpha)$.
Here, the indicator $\xi_t$ is carefully chosen so that
\begin{align*}
\bbE \left[\xi_t\right] = G(\varepsilon_t(\alpha)),
\end{align*}
and $\forall \alpha \in [\underline{\alpha}, 1]$, $0 < \varepsilon_t(\alpha) < 1 + W$.

\begin{restatable}{theorem}{thirdtheorem}
\label{thm: unknown alpha unknown G full}
Suppose that \Cref{asm: unknown alpha unknown G} holds. For repeated non-credible second-price auctions with unknown credibility $\alpha_0$, known distribution $G$ and full feedback, there exists a bidding algorithm $\Cref{algo: unknown alpha unknown G}$ that achieves $\widetilde{O}(T^{1/2})$ if \Cref{asm: unknown alpha unknown G 2} holds. And the lower bound on regret for this problem is $\Omega(T^{1/2})$.
\end{restatable}

We first give a bound on distance between $\widetilde{\alpha}_s$ and $\alpha_0$ in \Cref{lem: C.3}. Next we show that the distribution estimator $\hat{G}_s(\cdot; \widetilde{\alpha}_s)$ uniformly converges to the real distribution $G$ in \Cref{lem: C.4}. Our algorithm and the regret analysis follow the same spirit as in Theorem 4.3 in \citet{badanidiyuru2023learning}. The main difference is that we use $\xi_t$ rather than $x_t$ as our indicator function. We defer the proof of this section to \Cref{app: unknown alpha unknown G}.

\section{Conclusion}

In this paper, we formulate and study the problem of learning against non-credible auctions and evaluate several algorithms under various information structures.

It remains an open problem whether we can further relax the assumptions over distribution $G$, especially for \Cref{thm: unknown alpha known G} and \Cref{thm: unknown alpha unknown G full}. Can we design a near optimal algorithm for general distributions with unknown $\alpha_0$ and known $G$? 
Can we apply the upper bound result of \Cref{thm: unknown alpha unknown G full} to the bandit feedback model?
In this work, we assume the seller's credibility $\alpha_0$ is fixed. Another research direction is to consider that the seller may use a time-dependent charging rule. Will learning learning become impossible in this more complicated setting?
\bibliography{aaai24}

\begin{thebibliography}{29}
\providecommand{\natexlab}[1]{#1}

\bibitem[{Akbarpour and Li(2020)}]{akbarpour2020credible}
Akbarpour, M.; and Li, S. 2020.
\newblock Credible auctions: A trilemma.
\newblock \emph{Econometrica}, 88(2): 425--467.

\bibitem[{An(1996)}]{an1996log}
An, M. 1996.
\newblock Log-concave Probability Distributions: Theory and Statistical
  Testing.
\newblock Game theory and information, University Library of Munich, Germany.

\bibitem[{Badanidiyuru, Feng, and Guruganesh(2023)}]{badanidiyuru2023learning}
Badanidiyuru, A.; Feng, Z.; and Guruganesh, G. 2023.
\newblock Learning to Bid in Contextual First Price Auctions.
\newblock In \emph{Proceedings of the ACM Web Conference 2023}, 3489--3497.

\bibitem[{Balseiro et~al.(2019)Balseiro, Golrezaei, Mahdian, Mirrokni, and
  Schneider}]{balseiro2019contextual}
Balseiro, S.; Golrezaei, N.; Mahdian, M.; Mirrokni, V.; and Schneider, J. 2019.
\newblock Contextual bandits with cross-learning.
\newblock \emph{Advances in Neural Information Processing Systems}, 32.

\bibitem[{Balseiro, Besbes, and Weintraub(2015)}]{balseiro2015repeated}
Balseiro, S.~R.; Besbes, O.; and Weintraub, G.~Y. 2015.
\newblock Repeated auctions with budgets in ad exchanges: Approximations and
  design.
\newblock \emph{Management Science}, 61(4): 864--884.

\bibitem[{Balseiro and Gur(2019)}]{balseiro2019learning}
Balseiro, S.~R.; and Gur, Y. 2019.
\newblock Learning in repeated auctions with budgets: Regret minimization and
  equilibrium.
\newblock \emph{Management Science}, 65(9): 3952--3968.

\bibitem[{Balseiro, Lu, and Mirrokni(2022)}]{balseiro2022best}
Balseiro, S.~R.; Lu, H.; and Mirrokni, V. 2022.
\newblock The best of many worlds: Dual mirror descent for online allocation
  problems.
\newblock \emph{Operations Research}.

\bibitem[{Castiglioni, Celli, and Kroer(2022)}]{castiglioni2022online}
Castiglioni, M.; Celli, A.; and Kroer, C. 2022.
\newblock Online learning with knapsacks: the best of both worlds.
\newblock In \emph{International Conference on Machine Learning}, 2767--2783.
  PMLR.

\bibitem[{Cesa-Bianchi, Gentile, and Mansour(2014)}]{cesa2014regret}
Cesa-Bianchi, N.; Gentile, C.; and Mansour, Y. 2014.
\newblock Regret minimization for reserve prices in second-price auctions.
\newblock \emph{IEEE Transactions on Information Theory}, 61(1): 549--564.

\bibitem[{Chen et~al.(2023)Chen, Wang, Duan, Sun, Chen, Yan, and
  Deng}]{chen2023coordinated}
Chen, Y.; Wang, Q.; Duan, Z.; Sun, H.; Chen, Z.; Yan, X.; and Deng, X. 2023.
\newblock Coordinated Dynamic Bidding in Repeated Second-Price Auctions with
  Budgets.
\newblock \emph{arXiv preprint arXiv:2306.07709}.

\bibitem[{Chen et~al.(2022)Chen, Wang, Wang, Pan, Shi, Tang, Cai, Ren, Zhu, and
  Deng}]{chen2022dynamic}
Chen, Z.; Wang, C.; Wang, Q.; Pan, Y.; Shi, Z.; Tang, C.; Cai, Z.; Ren, Y.;
  Zhu, Z.; and Deng, X. 2022.
\newblock Dynamic Budget Throttling in Repeated Second-Price Auctions.
\newblock \emph{arXiv preprint arXiv:2207.04690}.

\bibitem[{Dianat and Freer(2021)}]{dianat2021credibility}
Dianat, A.; and Freer, M. 2021.
\newblock Credibility in Second-Price Auctions: An Experimental Test.
\newblock \emph{arXiv preprint arXiv:2105.00204}.

\bibitem[{D{\"u}mbgen and Rufibach(2009)}]{dumbgen2009maximum}
D{\"u}mbgen, L.; and Rufibach, K. 2009.
\newblock Maximum likelihood estimation of a log-concave density and its
  distribution function: Basic properties and uniform consistency.
\newblock \emph{Bernoulli}, 15(1): 40--68.

\bibitem[{Feng, Padmanabhan, and Wang(2022)}]{feng2022online}
Feng, Z.; Padmanabhan, S.; and Wang, D. 2022.
\newblock Online Bidding Algorithms for Return-on-Spend Constrained
  Advertisers.
\newblock \emph{arXiv preprint arXiv:2208.13713}.

\bibitem[{Golrezaei et~al.(2021)Golrezaei, Jaillet, Liang, and
  Mirrokni}]{golrezaei2021bidding}
Golrezaei, N.; Jaillet, P.; Liang, J. C.~N.; and Mirrokni, V. 2021.
\newblock Bidding and pricing in budget and ROI constrained markets.
\newblock \emph{arXiv preprint arXiv:2107.07725}.

\bibitem[{Hakimov and Raghavan(2023)}]{epub94449}
Hakimov, R.; and Raghavan, M. 2023.
\newblock Improving Transparency and Verifiability in School Admissions: Theory
  and Experiment.

\bibitem[{Han et~al.(2020)Han, Zhou, Flores, Ordentlich, and
  Weissman}]{han2020learning}
Han, Y.; Zhou, Z.; Flores, A.; Ordentlich, E.; and Weissman, T. 2020.
\newblock Learning to bid optimally and efficiently in adversarial first-price
  auctions.
\newblock \emph{arXiv preprint arXiv:2007.04568}.

\bibitem[{Han, Zhou, and Weissman(2020)}]{han2020optimal}
Han, Y.; Zhou, Z.; and Weissman, T. 2020.
\newblock Optimal no-regret learning in repeated first-price auctions.
\newblock \emph{arXiv preprint arXiv:2003.09795}.

\bibitem[{IAB(2023)}]{iab2022report}
IAB, T. I. A.~B. 2023.
\newblock Internet Advertising Revenue Report, Full year 2022 results.

\bibitem[{Iyer, Johari, and Sundararajan(2014)}]{iyer2014mean}
Iyer, K.; Johari, R.; and Sundararajan, M. 2014.
\newblock Mean field equilibria of dynamic auctions with learning.
\newblock \emph{Management Science}, 60(12): 2949--2970.

\bibitem[{Javanmard and Nazerzadeh(2019)}]{javanmard2019dynamic}
Javanmard, A.; and Nazerzadeh, H. 2019.
\newblock Dynamic pricing in high-dimensions.
\newblock \emph{The Journal of Machine Learning Research}, 20(1): 315--363.

\bibitem[{McAdams and Schwarz(2007)}]{mcadams2007pays}
McAdams, D.; and Schwarz, M. 2007.
\newblock Who pays when auction rules are bent?
\newblock \emph{International Journal of Industrial Organization}, 25(5):
  1144--1157.

\bibitem[{Porter and Shoham(2005)}]{porter2005cheating}
Porter, R.; and Shoham, Y. 2005.
\newblock On cheating in sealed-bid auctions.
\newblock \emph{Decision Support Systems}, 39(1): 41--54.

\bibitem[{Rothkopf and Harstad(1995)}]{rothkopf1995two}
Rothkopf, M.~H.; and Harstad, R.~M. 1995.
\newblock Two models of bid-taker cheating in Vickrey auctions.
\newblock \emph{Journal of Business}, 257--267.

\bibitem[{Tsybakov(2009)}]{tsybakov2009introduction}
Tsybakov, A.~B. 2009.
\newblock \emph{Introduction to Nonparametric Estimation}.
\newblock New York: Springer-Verlag.

\bibitem[{Vickrey(1961)}]{vickrey1961counterspeculation}
Vickrey, W. 1961.
\newblock Counterspeculation, auctions, and competitive sealed tenders.
\newblock \emph{The Journal of finance}, 16(1): 8--37.

\bibitem[{Wang et~al.(2023)Wang, Yang, Deng, and Kong}]{wang2023learning}
Wang, Q.; Yang, Z.; Deng, X.; and Kong, Y. 2023.
\newblock Learning to bid in repeated first-price auctions with budgets.
\newblock \emph{arXiv preprint arXiv:2304.13477}.

\bibitem[{Weed, Perchet, and Rigollet(2016)}]{weed2016online}
Weed, J.; Perchet, V.; and Rigollet, P. 2016.
\newblock Online learning in repeated auctions.
\newblock \emph{Journal of Machine Learning Research}, 49(June): 1562--1583.

\bibitem[{Zhang et~al.(2022)Zhang, Han, Zhou, Flores, and
  Weissman}]{zhang2022leveraging}
Zhang, W.; Han, Y.; Zhou, Z.; Flores, A.; and Weissman, T. 2022.
\newblock Leveraging the Hints: Adaptive Bidding in Repeated First-Price
  Auctions.
\newblock \emph{arXiv preprint arXiv:2211.06358}.

\end{thebibliography}

\newpage
\appendix
\onecolumn

{
\linespread{1.2}
\section*{\centering \Large Supplementary Material of ``Learning against Non-credible Auctions''}
}
\vspace{1em}

\section{Missing Algorithms and Proofs of \Cref{sec: known alpha unknown G}}
\label{app: known alpha unknown G}

\begin{algorithm}
\caption{Learning with known $\alpha_0 = 0$ (or unknown $\alpha_0$), unknown $G$ and bandit feedback}
\label{algo: known alpha unknown G fpa}
    \SetKwInput{KwInit}{Initialization}
    \KwIn{Time horizon $T$.}
    \KwInit{Let $\calB = \{b^1, \cdots, b^K\}$ be the bid set with $b^k = (k-1)/K$; for $t \in [K]$, the bidder submits a bid $b_t = b^t$.}
    \For{$t \gets K + 1$ \KwTo $T$}{
        The bidder receives the value $v_t \in [0, 1]$. \\
        For every $k\in [K]$, the bidder counts the observation by
        \begin{align}
            n_t^k \gets \sum_{s=1}^{t-1} \mathbb{I}\left\{b_s = b^k\right\}.
        \end{align}\\
        For every $k\in [K]$, the bidder computes the confidence bound:
        \begin{align}
            w_t^k = \sqrt{\frac{2\ln{T}}{n_t^k}},
        \end{align}\\
        For every $k \in [K]$, the bidder estimates the reward by
        \begin{align}
            \widetilde{r}_t(v_t, b^k) = &\frac{1}{n_t^k}\sum_{s = 1}^{t-1}\mathbb{I}\left\{b_s = b^k\right\}\left(x_s v_t - c_s\right),
        \end{align}\\
        The bidder submits a bid $b_t = \arg\max_{b\in \calB} \left( \widetilde{r}_t(v_t, b) + w_t^k\right)$.
    }
\end{algorithm}

\begin{algorithm}[!ht]
\caption{Learning with known $\alpha_0\in (0,1)$, unknown $G$ and bandit feedback}
\label{algo: known alpha unknown G main}
    \SetKwInput{KwInit}{Initialization}
    \KwIn{Time horizon $T$; credibility parameter $\alpha_0 \in (0, 1)$.}
    \KwInit{Let $\calV = [v^1, \cdots, v^M]$ be the value set with $v^m = (m-1)/M$;
        let $\calB = \{b^1, \cdots, b^K\}$ be the bid set with $b^k = (k-1)/K$; 
        set $\calB_1^m \gets \calB$ for each $v_m \in \calV$; 
        select failure probability $\delta\in (0,1)$.}
    \For{$t \gets 1$ \KwTo $T$}{
        The bidder receives the value $v_t \in [0, 1]$ and round it to $v^{m(t)}$ by
        \begin{align}
            v^{m(t)} = \max\{u\in \calV: u \leq v_t\}.
        \end{align}\\
        The bidder submits a bid $b_t = \sup \calB_{t}^{m(t)}$. \\
        For every $k\in [K]$, the bidder counts the observation by
        \begin{align}
            n_t^k \gets \sum_{s=1}^{t} \mathbb{I}\left\{b_s \geq b^k\right\}.
        \end{align}\\
        For every $m \in [M], k \in [K]$, the bidder estimates the reward by
        \begin{align}
        \label{eqn: reward estimation}
            \widetilde{r}_t(v^m, b^k, \alpha_0) = &\frac{1}{n_t^k}\sum_{s = 1}^{t}\mathbb{I}\left\{b_s \geq b^k\right\}\mathbb{I}\left\{b^k \geq \widetilde{d}_s\right\}x_s\left(v^m - \alpha_0\widetilde{d}_s -(1-\alpha_0) b^k\right),
        \end{align}
        where $\widetilde{d}_s = \left(c_s - (1-\alpha_0)b_s\right)/\alpha_0$.\\
        \For{$m = M,M-1,\ldots, 1$}{
            The bidder eliminates bids by:
            \begin{align}
                \calB_{t}^m = \left\{b^k\in \calB_{t}^m: b^k \leq \min_{s > m} \sup \calB_{t+1}^{s}\right\}.
            \end{align} \\
            The bidder computes the confidence bound:
            \begin{align}
                w_t^m = \sqrt{\frac{4\ln{T}\log(KT/\delta)}{N_t^m}},
            \end{align}
            where $N_t^m = \min_{b^k\in \calB_{t}^m} n_t^k$. \\
            The bidder eliminates bids by:
                \begin{align}
                    \calB_{t+1}^m \gets \Big\{& b^k\in  \calB_{t}^m: \widetilde{r}_t(v^m, b^k, \alpha_0) \geq \max_{b \in \calB_{t}^m}\widetilde{r}_t(v^m, b, \alpha_0) - 2 w_t^m \Big\}. 
                \end{align}
        }
    }
\end{algorithm}

\increasing*

\begin{proof}[Proof of \Cref{lem: optimal bid monotonicity}]
Fix any $v_1, v_2$ with $v_1 \leq v_2$. For any $b \leq b^*(v_1, \alpha)$, it holds that
\begin{align*}
    r(v_2, b, \alpha) & = r(v_1, b, \alpha) + (v_2 - v_1)G(b) \\
    & \leq r(v_1, b^*(v_1, \alpha), \alpha) + (v_2 - v_1)G(b^*(v_1, \alpha)) \\
    & = r(v_2, b^*(v_1, \alpha), \alpha),
\end{align*}
where the inequality follows from the definition of $b^*(v_1, \alpha)$ and the conditions $v_1 \leq v_2, b \leq b^*(v_1, \alpha)$. Thus, all bids no larger than $b^*(v_1, \alpha)$ cannot be the largest maximizer for $v_2$.

Fix any $\alpha_1, \alpha_2$ with $\alpha_1 \leq \alpha_2$. For any $b \leq b^*(v, \alpha_1)$, it holds that
\begin{align*}
    r(v, b, \alpha_2) & = r(v, b, \alpha_1) + (\alpha_2 - \alpha_1) \int_{0}^{b} G(y)\,dy \\
    & \leq r(v, b^*(v, \alpha_1), \alpha_1) + (\alpha_2 - \alpha_1) \int_{0}^{b^*(v, \alpha_1)} G(y)\,dy\\
    & = r(v, b^*(v, \alpha_1), \alpha_2),
\end{align*}
where the inequality follows from the definition of $b^*(v, \alpha)$ and the conditions $\alpha_1 \leq \alpha_2, b \leq b \leq b^*(v, \alpha_1)$. Thus, all bids no larger than $b^*(v, \alpha_1)$ cannot be the largest maximizer for $\alpha_2$.
\end{proof}

\knownalphalower*

\begin{proof}[Proof of \Cref{lem: known alpha unknown G main lower}]
The proof follows the Le Cam’s two-point method \cite{tsybakov2009introduction}.
For any $\alpha_0 \in (0,1)$, consider the following two candidate distributions supported on $[0,1]$:
\begin{equation*}
G_1(x) = \begin{cases}
0 & \text{if } x \in \left[0,\frac{1-\alpha_0}{3-2\alpha_0}\right) \\
\frac{1}{2} + \Delta & \text{if } x \in \left[\frac{1-\alpha_0}{3-2\alpha_0}, \frac{2-\alpha_0}{3-2\alpha_0}\right) \\
1 & \text{if } x \in \left[\frac{2-\alpha_0}{3-2\alpha_0}, 1\right]
\end{cases}, \qquad
G_2(x) = \begin{cases}
0 & \text{if } x \in \left[0,\frac{1-\alpha_0}{3-2\alpha_0}\right) \\
\frac{1}{2} - \Delta & \text{if } x \in \left[\frac{1-\alpha_0}{3-2\alpha_0}, \frac{2-\alpha_0}{3-2\alpha_0}\right) \\
1 & \text{if } x \in \left[\frac{2-\alpha_0}{3-2\alpha_0}, 1\right]
\end{cases},
\end{equation*}
where $\Delta$ is some parameter to be chosen later. In other words, $G_1$ corresponds to a discrete random variable taking value in $\left\{\frac{1-\alpha_0}{3-2\alpha_0}, \frac{2-\alpha_0}{3-2\alpha_0}\right\}$ with probability $(1/2+\Delta, 1/2-\Delta)$, and $G_2$ corresponds to the probability $(1/2-\Delta, 1/2+\Delta)$. Fixing $v_t \equiv 1$, let $R_1(b)$ and $R_2(b)$ be the expected per-round reward when bidding $b$ under $G_1$ and $G_2$, respectively. After some algebra, it is straightforward to check that
\begin{align*}
\max_{b\in [0,1]} R_1(b) &= \max_{b\in [0,1]} \left[(1-b)G_1(b) + \alpha_0 \int_{0}^b G_1(x)\mathrm{d}x\right] = \frac{2-\alpha_0}{3-2\alpha_0}\left(\frac{1}{2}+\Delta\right),\\
\max_{b\in [0,1]} R_2(b) &= \max_{b\in [0,1]} \left[(1-b)G_2(b) + \alpha_0 \int_{0}^b G_2(x)\mathrm{d}x\right] = \frac{1}{3-2\alpha_0}\left(1-\frac{\alpha_0}{2}+\alpha_0\Delta\right),\\
\max_{b\in [0,1]} \left(R_1(b)+R_2(b)\right) &= \max_{b\in [0,1]} \left[(1-b)\left(G_1(b)+G_2(b)\right) + \alpha_0 \int_{0}^b \left(G_1(x)+G_2(x)\right)\mathrm{d}x\right] = \frac{2-\alpha_0}{3-2\alpha_0}.
\end{align*}
Hence, for any $b_t\in[0,1]$, we have
\begin{align}
&\left(\max_{b\in[0,1]}R_1(b) - R_1(b_t)\right) + \left(\max_{b\in[0,1]}R_2(b) - R_2(b_t)\right) \nonumber \\
&\geq \max_{b\in[0,1]}R_1(b) + \max_{b\in[0,1]} R_2(b) -\max_{b\in[0,1]}\left(R_1(b) + R_2(b)\right) \nonumber \\
&=\frac{2-2\alpha_0}{3-2\alpha_0} \Delta. \label{eq: two-point method}
\end{align}
The inequality \eqref{eq: two-point method} is the separation condition required in the two-point method: there is no single bid $b_t$ that can obtain a uniformly small instantaneous regret under both $G_1$ and $G_2$.

In the full information model, when $\alpha_0 \neq 0$, the bidder can always infer $d_t$ by
\begin{align*}
    d_t = \left(p_t - (1-\alpha_0)b_t\right)/\alpha_0.
\end{align*}
Let $P_i^t, i\in \{1, 2\}$ be the distribution of all historical samples $(d_1, \cdots, d_{t-1})$ at the beginning of time $t$. Then for any policy $\pi$, 
\begin{align}
\sup_{G} \mathrm{Regret}(\pi; G) & \overset{(\text a)}{\geq} \frac{1}{2}\mathrm{Regret}(\pi; G_1) + \frac{1}{2}\mathrm{Regret}(\pi; G_2)\nonumber \\
&=\frac{1}{2}\sum_{t=1}^T \left(\mathbb{E}_{P_1^t}\left[\max_{b\in[0,1]} R_1(b) - R_1(b_t)\right]+\mathbb{E}_{P_2^t}\left[\max_{b\in[0,1]} R_2(b) - R_2(b_t)\right]\right) \nonumber \\
&\overset{(\text b)}{\geq} \frac{1}{2}\sum_{t=1}^T \frac{2-2\alpha_0}{3-2\alpha_0}\Delta\int \min \{\mathrm{d}P_1^t, \mathrm{d}P_2^t\} \nonumber \\
&\overset{(\text c)}{=} \frac{1}{2}\sum_{t=1}^T \frac{2-2\alpha_0}{3-2\alpha_0}\Delta \left(1-\|P_1^t-P_2^t\|_{\mathrm{TV}}\right) \nonumber \\
&\geq \frac{1-\alpha_0}{3}\Delta \sum_{t=1}^T \left(1-\|P_1^t-P_2^t\|_{\mathrm{TV}}\right), \label{eq: total variation distance argument}
\end{align}
where (a) is due to the fact that the maximum is no smaller than the average, (b) follows from \eqref{eq: two-point method}, and (c) is due to the identity $\int \min \{\mathrm{d}P, \mathrm{d}Q\} = 1 - \|P-Q\|_{\mathrm{TV}}$. Invoking \Cref{lem: tvd-kl} and using the fact that for $\Delta \in (0,1/4)$, 
\begin{align}
D_{\mathrm{KL}}(P_1^t || P_2^t) &= (t-1) D_{\mathrm{KL}}(G_1 || G_2) \nonumber \\
&= (t-1) \left(\left(\frac{1}{2}+\Delta\right)\log\frac{1/2+\Delta}{1/2-\Delta}+\left(\frac{1}{2}-\Delta\right)\log\frac{1/2-\Delta}{1/2+\Delta} \right) \nonumber \\
&\leq 16T\Delta^2,  \label{eq: kl decomposition}
\end{align}
we have the following inequality on total variation distance:
\begin{equation}
\label{eq: total variation distance inequality}
1 - \|P_1^t -P_2^t\|_{\mathrm{TV}} \geq \frac{1}{2}\exp\left(-16T\Delta^2\right).
\end{equation}
Finally, choosing $\Delta = \frac{1}{4\sqrt{T}}$, combining inequalities \eqref{eq: total variation distance argument} and \eqref{eq: total variation distance inequality}, we conclude that \Cref{lem: known alpha unknown G main lower} holds with the constant $c = \frac{1}{24e}$.
\end{proof}

\begin{lemma}[\citet{tsybakov2009introduction}]
\label{lem: tvd-kl}
    Let $P, Q$ be two probability measures on the same space. It holds that
    \begin{equation*}
        1-\|P-Q\|_{\mathrm{TV}}\geq \frac{1}{2}\exp\left(-\frac{D_{\mathrm{KL}}(P||Q)+D_{\mathrm{KL}}(Q||P)}{2}\right).
    \end{equation*}
\end{lemma}
\section{Missing Proofs of \Cref{sec: unknown alpha known G}}
\label{app: unknown alpha known G}

\alphaconvergence*

\begin{proof}[Proof of \Cref{lem: convergence of alpha unknown_alpha_known_G}]
Taking expectation on $r_s$ w.r.t. $d_s$, one has
\begin{align*}
    \bbE \left[r_s\right] = r(v_s, b_s, \alpha_0).
\end{align*}
Then by applying Azuma–Hoeffding inequality, with probability at least $1 - \delta/T$,
\begin{align*}
    \left|\sum_{s=1}^{t-1} \left(r_s - r(v_s, b_s, \alpha_0)\right)\right| \leq \sqrt{2(t-1)\log{(2T/\delta)}}.
\end{align*}
Then by \eqref{eqn: estimate alpha  unknown_alpha_known_G} in \Cref{algo: unknown alpha known G},
\begin{align*}
    \left|\sum_{s=1}^{t-1} \left(\widetilde{\alpha}_t - \alpha_0\right) \int_{0}^{b_s}G(y)\, dy\right| \leq & \left|\sum_{s=1}^{t-1} \left(r_s - r(v_s, b_s, \widetilde{\alpha}_t)\right)\right| + \left|\sum_{s=1}^{t-1} \left(r_s - r(v_s, b_s, \alpha_0)\right)\right| \\
   \leq & 2 \left|\sum_{s=1}^{t-1} \left(r_s - r(v_s, b_s, \alpha_0)\right)\right| \\
   \leq & 2\sqrt{2(t-1)\log{(2T/\delta)}}.
\end{align*}
Then the proof can be concluded by applying a union bound.
\end{proof}

\secondtheorem*

\begin{proof}[Proof of \Cref{thm: unknown alpha known G}]
First, we show that the loss incurred by choosing such a bid $b_t$ is small with probability at least $1-\delta$. For all $t \geq 2$, 
\begin{align}
    r(v_t, b_t, \alpha_0) & = r(v_t, b_t, \widetilde{\alpha}_t) - \left(\widetilde{\alpha}_t - \alpha_0 \right)\int_{0}^{b_t}G(y)\,dy, \nonumber \\
    & \overset{(\text a)}{\geq} r(v_t, b_t^*, \widetilde{\alpha}_t) - \left(\widetilde{\alpha}_t - \alpha_0 \right)\int_{0}^{b_t}G(y)\,dy, \nonumber \\
    & = r(v_t, b_t^*, \alpha_0) + \left(\widetilde{\alpha}_t - \alpha_0 \right)\int_{b_t}^{b_t^*}G(y)\,dy, \nonumber \\
    & \overset{(\text b)}{\geq} r(v_t, b_t^*, \alpha_0) - \frac{1}{B_1}(\widetilde{\alpha}_t - \alpha_0)^2, \nonumber \\
    & \overset{(\text c)}{\geq} r(v_t, b_t^*, \alpha_0) - \frac{1}{B_1}w_t^2, \label{eqn: near to optimal unknown_alpha_known_G}
\end{align}
where (a) holds by the choice of $b_t$ and (b) holds since $b_t^* = b^*(v_t, \alpha_0)$, $b_t = b^*(v_t, \widetilde{\alpha}_t)$ together with \Cref{lem: partial derivative upper bound}; (c) follows from \Cref{lem: convergence of alpha unknown_alpha_known_G}.

Next, we have
\begin{align}
    \sum_{s=1}^{t-1} \int_{0}^{b_s}G(y)dy & = \sum_{s=1}^{t-1} \int_{0}^{b^*(v_s, \widetilde{\alpha}_s)}G(y)\,dy. \nonumber \\
    & \overset{(\text a)}{\geq} \sum_{s=1}^{t-1} \int_{0}^{b^*(v_s, 0)}G(y)\,dy. \nonumber \\
    & \overset{(\text b)}{\geq} (t-1)\cdot \underbrace{\bbE_{v} \left[\int_{0}^{b^*(v, 0)}G(y)\,dy\right]}_{C_0} - \sqrt{2(t-1)\log{(T/\delta)}}, \nonumber
\end{align}
where (a) holds by \Cref{lem: optimal bid monotonicity} or \Cref{lem: partial derivative upper bound} and (b) holds with probability at least $1 - \delta/T$ by applying Azuma–Hoeffding inequality.

By \Cref{lem: C_0 constant unknown_alpha_known_G}, If $C_0 = 0$, no algorithm can obtain positive expected reward under regardless of the seller's credibility $\alpha_0$, which also implies any algorithm that guarantees $b_t \leq v_t$ can get a zero regret. Therefore, we can presume $C_0$ is a positive constant depending on the known distribution $G$. Then for some constant $C_1 \in (0, C_0)$, when $t > t_0 =  \frac{2\log{(T/\delta)}}{(C_0 - C_1)^2}$, we have 
\begin{align}
    w_t & \leq \frac{2\sqrt{2(t-1)\log{(2T/\delta)}}}{C_0(t-1) - \sqrt{2(t-1)\log{(T/\delta)}}} \nonumber \\
    & \leq \frac{2\sqrt{2\log{(2T/\delta)}}}{C_0\sqrt{t-1} - \sqrt{2\log{(T/\delta)}}} \nonumber \\
    & \leq \frac{2\sqrt{2\log{(2T/\delta)}}}{C_1\sqrt{t-1}}. \label{eqn: w upper bound unknown_alpha_known_G}
\end{align}

Combining \eqref{eqn: near to optimal unknown_alpha_known_G} and \eqref{eqn: w upper bound unknown_alpha_known_G}, we have with probability at least $1-2\delta$, 
\begin{align*}
    \sum_{t=1}^T r(v_t, b_t^*, \alpha_0) - \sum_{t=1}^T r(v_t, b_t, \alpha_0) \leq & t_0 + \frac{1}{B_1}\sum_{t=t_0+1}^{T} w_t^2 \\
    \leq & t_0 + \frac{8\log{(2T/\delta)}}{B_1C_1^2}  \sum_{t=t_0+1}^{T} \frac{1}{t-1} \\
    \leq & \frac{2\log{(T/\delta)}}{(C_0 - C_1)^2} + \frac{8\log{(2T/\delta)}}{B_1C_1^2} \log{T},
\end{align*}
which is $O(\log^2{T})$ by taking $\delta \sim T^{-1}$. The proof of the regret upper bound can be finished by taking expectation.

The lower bound part trivially holds. We have for some constant $c$,
\begin{align*}
    \inf_{\pi} \sup_{\alpha_0} \mathrm{Regret}(\pi) \geq c,
\end{align*}
since at least in the first round, no single bid $b_t$ that can obtain a uniformly small instantaneous regret for all $\alpha_0$.
\end{proof}

\secondcorollary*

\begin{proof}[Proof of \Cref{cor: unknown alpha known G}]
The proof is analogous to that of \Cref{thm: unknown alpha known G}. We first have
\begin{align}
    r(v_t, b_t, \alpha_0) & = r(v_t, b_t, \widetilde{\alpha}_t) - \left(\widetilde{\alpha}_t - \alpha_0 \right)\int_{0}^{b_t}G(y)\,dy, \nonumber \\
    & \geq r(v_t, b_t^*, \widetilde{\alpha}_t) - \left(\widetilde{\alpha}_t - \alpha_0 \right)\int_{0}^{b_t}G(y)\,dy, \nonumber \\
    & = r(v_t, b_t^*, \alpha_0) + \left(\widetilde{\alpha}_t - \alpha_0 \right)\int_{b_t}^{b_t^*}G(y)\,dy, \nonumber \\
    & \geq r(v_t, b_t^*, \alpha_0) - w_t, \label{eqn: near to optimal 2 unknown_alpha_known_G}
\end{align}
where the second inequality uses the face $\int_{b_t}^{b_t^*}G(y)\,dy \leq 1$ rather than \Cref{lem: partial derivative upper bound}. Together the inequality \eqref{eqn: w upper bound unknown_alpha_known_G}, which still holds given that $G$ is continuous, the proof can be finished by bounding $\sum_{t=1}^T w_t$.
\end{proof}

\begin{lemma}
\label{lem: partial derivative upper bound}
Suppose that \Cref{asm: unknown alpha known G} holds. Then $b^*(v, \alpha)$ is strictly increasing in $\alpha$ with derivative bounded by $1/B_1$.
\end{lemma}

\begin{proof}[Proof of \Cref{lem: partial derivative upper bound}]
First, we have
\begin{align*}
    \frac{\partial r(v, b, \alpha)}{\partial b} = 0 \Longrightarrow 1 - (v-b)\frac{g(b)}{G(b)} = \alpha.
\end{align*}
Denote $\phi(b) = 1 - (v-b)\frac{g(b)}{G(b)}$. Give that \Cref{asm: unknown alpha known G} holds, we have $\phi'(b) = \log'G(b) - (v-b)\log''G(b) \geq B_1$. Thus, $\phi(b)$ is a strictly increasing function and we have
\begin{align*}
    \frac{\partial b^*(v, \alpha)}{\partial \alpha} = (\phi^{-1}(\alpha))' = \frac{1}{\phi'(\phi^{-1}(\alpha))} \leq \frac{1}{B_1}.
\end{align*}
\end{proof}

\begin{lemma}
\label{lem: C_0 constant unknown_alpha_known_G}
Suppose that $G$ is continuous. If $\bbE_v \left[\int_0^{b^*(v, 0)} G(y)\,dy\right] = 0$, then for any $\alpha_0$, no algorithm can obtain positive expected reward.
\end{lemma}

\begin{proof}[Proof of \Cref{lem: C_0 constant unknown_alpha_known_G}]
First, as $b^*(v, 0)$ is increasing in $v$ by \Cref{lem: optimal bid monotonicity}, $\int_0^{b^*(v, 0)} G(y)\,dy$ is also increasing in $v$. Thus, if $\bbE_v \left[\int_0^{b^*(v, 0)} G(y)\,dy\right] = 0$, it holds that for nearly all $v$, 
\begin{align*}
    \int_0^{b^*(v, 0)} G(y)\,dy = 0, 
\end{align*}
except a set contained in $[F^{-1}(1), 1]$ on which the probability distribution $F$ has a zero measure.

Next, by the continuity and monotonicity of $G$, for all $v \in [0, F^{-1}(1))$, we have $G(y) \equiv 0$ on $[0, b^*(v, 0)]$. As
$$r(v, b^*(v, 0), 0) = (v - b^*(v, 0)) G(b^*(v, 0)) = 0 = r(v, v, 0),$$ we have $b^*(v, 0) \geq v$ by the tie-breaking rule in the definition of $b^*(v, 0)$. Thus, $G(y)$ remains zero on $[0, v]$ for all $v < F^{-1}(1)$. Again by the continuity of $G$, we have $G(y) \equiv 0$ on $[0, F^{-1}(1)]$ where $F^{-1}(1) = \inf \{v: F(v) = 1\}$.

For nearly all $v$ except a zero-measure set, 
\begin{align*}
    r(v, b, \alpha_0) & = (v -b)G(b) + \alpha \int_{0}^{b}G(y)\, dy \\
    & = \alpha_0 \int_{v}^{b} \alpha_0 G(y) - G(v) \, dy \\
    & \leq 0,
\end{align*}
where the inequality holds (1) by $G(y) \equiv 0$ when $b \leq v$ and (2) by $\alpha G(y) \leq G(v)$ for $y \in [v, b]$ when $b \geq v$. Therefore, no algorithm can obtain positive expected reward under such a distribution $G$ regardless of the seller's credibility $\alpha_0$.
\end{proof}

\begin{example}
\label{eg: bad example unknown_alpha_known_G} Let $v \equiv 1$. Consider the following distribution:
\begin{equation*}
    G(y) = \begin{cases}
    0 & \text{if } 0 \leq y < \frac{1}{3} \\
    \frac{3}{4}y + \frac{1}{4} & \text{if } \frac{1}{3} \leq y \leq 1
    \end{cases}.
\end{equation*}
Then for $b \in [1/3, 1]$, 
\begin{align*}
    r(v, b, \alpha_0) & = (1-b)G(b) + \alpha_0 \int_{0}^{b}G(y)\,dy \\
    & = - \frac{6-3\alpha_0}{8}b^2 + \frac{2 + \alpha_0}{4}b + \frac{2 - \alpha_0}{8}.
\end{align*}
The optimal bid is calculated by
\begin{align*}
    b^* & = \arg\max_{b} r(v, b, \alpha_0) = \frac{1}{3}\left(1 + \frac{2}{2/\alpha_0 - 1}\right).
\end{align*}
Taking $\alpha_0 = \frac{1}{\sqrt{T} + 1/2}$, we have 
\begin{align*}
    \int_{0}^{b^*} G(y) dy & = \int_{\frac{1}{3}}^{\frac{1}{3} \cdot \left(1+\frac{1}{\sqrt{T}}\right)} \left(\frac{3}{4}y + \frac{1}{4}\right) dy = \frac{1}{24} (\frac{1}{T} + \frac{4}{\sqrt{T}}) \longrightarrow 0.
\end{align*}
Thus, $C_0$ in the proof of \Cref{thm: unknown alpha known G} is neither a positive constant nor $0$, but an infinitesimal o(1). As a result, we are not able to give an $O(\log{T}/\sqrt{t-1})$ upper bound for $w_t$.

Intuitively, it requires $b_t \rightarrow b_t^*$ for an algorithm to achieve low regret. However, when $b_t$ is approaching $\frac{1}{3}\left(1+\frac{1}{\sqrt{T}}\right)$, although the bidder can still win with positive probability, it becomes harder and harder for the bidder to distinguish different $\alpha_0$ under bandit feedback since in the winning rounds $b_t - d_t \leq b_t - \frac{1}{3} \longrightarrow 0$.
\end{example}
\section{Missing Proofs of \Cref{sec: unknown alpha unknown G}}
\label{app: unknown alpha unknown G}

\begin{proposition}
\label{pro: h_W and l_W}
Suppose that \Cref{asm: unknown alpha unknown G 2} holds.
There exists positive constants $l_W$ and $h_W$, such that $\forall x \in [0, 1 + W]$,
\begin{align}
    \max \{|\log'G(x)|, |\log'(1 - G(x))|\} & \leq h_W, \label{eq:h_W} \\
    \min \{-\log''G(x), -\log''(1 - G(x))\} & \geq l_W. \label{eq:l_W}
\end{align}
\end{proposition}

Note that if the density function $g$ is log-concave, then the cumulative distribution function $G$ and $1-G$ are both log-concave \cite{an1996log}. Thus, \Cref{eq:h_W} and \Cref{eq:l_W} trivially holds for the bounded interval $[0, 1 + W]$.

\convergencebound*

\begin{proof}[Proof of \Cref{lem: C.2}]
Denote $r_s = \frac{1}{2B_1\sqrt{T_s}}$. For any $d \in [0, 1 + W]$, the probability that there exists a point $y \in \{d_t\}_{t \in \Gamma_s}$ such that $|y - d| \leq r_s$ is at least
\begin{align*}
    1 - (1 - 2B_1 \cdot r_s)^{T_s} = 1 - \left(1 - \frac{1}{\sqrt{T_s}}\right)^{T_s} \approx 1 - e^{-\sqrt{T_s}} \geq 1 - \delta/2.
\end{align*}
Then for any $d$ and $y$ such that $|y - d| \leq r_s$, we can decompose $|\hat{G}_s(d;\alpha_0) - G(z)|$ in the following,
\begin{align*}
    &|\hat{G}_s(d;\alpha_0) - G(z)| \\ \leq \, &|\hat{G}_s(d;\alpha_0) - \hat{G}_s(y;\alpha_0)| + |\hat{G}_s(y;\alpha_0) - \mathbb{G}_s(y)| + |\mathbb{G}_s(y) - G(y)| + |G(y) - G(d)|
    \\ \leq \, &B_2 \cdot r_s + \frac{1}{T_s} + \sqrt{\frac{log(4 / \delta)}{2T_s}} + B_2 \cdot r_s
    \\ = \, &\frac{B_2}{B_1\sqrt{T_s}} + \sqrt{\frac{log(4 / \delta)}{2T_s}} + \frac{1}{T_s}.
\end{align*}
The second inequality follows from $\hat{G}_s(\cdot;\alpha_0)$ is $B_2$-Lipschitz, \Cref{lem: C.1},  Dvoretzky-Kiefer-Wolfowitz (DKW) inequality and $G$ is $B_2$-Lipschitz. It also holds with probability at least $1 - \delta/2$ since $T_s \gg log^2(2/\delta)$.
\end{proof}

Given \Cref{lem: C.2}, $\hat{G}_s(\cdot; \alpha_0)$ is arbitrarily close to $G$ when $T_s$ is sufficiently large. In addition, \citet{dumbgen2009maximum} also show $\hat{g}_s(\cdot, \alpha_0)$ is arbitrarily close to $g$ when $T_s$ is sufficiently large. Therefore, we can show,

\begin{proposition}
\label{pro: h_W_hat and l_W_hat}
Suppose that \Cref{asm: unknown alpha unknown G 2} holds. Under \Cref{algo: unknown alpha unknown G}, there exists positive constants $\Tilde{l}_W$ and $\Tilde{h}_W$, such that both
\begin{align*}
    max \{|\log'\hat{G}_s(x; \alpha_0)|, |\log'(1 - \hat{G}_s(x; \alpha_0))|\} \leq \Tilde{h}_W, \forall x \in [0, 1 + W], \forall s \in [S]
\end{align*}
and 
\begin{align*}
    min \{-\log''\hat{G}_s(x; \alpha), -\log''(1 - \hat{G}_s(x; \alpha))\} \geq \Tilde{l}_W, \forall x \in [0, 1 + W], \forall \alpha \in [\underline{\alpha}, 1], , \forall s \in [S]
\end{align*}
hold almost surely.
\end{proposition}

\thirdtheorem*

\begin{proof}[Proof of \Cref{thm: unknown alpha unknown G full}]
The lower bound directly follows from \Cref{lem: known alpha unknown G main lower} since no algorithm can avoid $\Omega(\sqrt{T})$ regret even with known $\alpha_0$. Then we mainly focus on the upper bound. We first rewrite the regret per round in the following way,
\begin{align*}
    & r(v_t, b_t^*, \alpha_0) - r(v_t, b_t, \alpha_0) \\
    = &\int_0^{b_t^*}(v_t - \alpha_0x - (1 - \alpha_0)b_t^*)g(x)dx - \int_0^{b_t}(v_t - \alpha_0x - (1 - \alpha_0)b_t)g(x)dx\\
    = &\int_0^{b_t^*}(v_t - \alpha_0x - (1 - \alpha_0)b_t^*)g(x)dx - \int_0^{b_t^*}(v_t - \alpha_0x - (1 - \alpha_0)b_t^*)\hat{g}_{s-1}(x; \widetilde{\alpha}_{s-1})dx \\
    &+ \int_0^{b_t^*}(v_t - \alpha_0x - (1 - \alpha_0)b_t^*)\hat{g}_{s-1}(x; \widetilde{\alpha}_{s-1})dx - \int_0^{b_t}(v_t - \alpha_0x - (1 - \alpha_0)b_t)\hat{g}_{s-1}(x; \widetilde{\alpha}_{s-1})dx \\
    &+ \int_0^{b_t}(v_t - \alpha_0x - (1 - \alpha_0)b_t)\hat{g}_{s-1}(x; \widetilde{\alpha}_{s-1})dx - \int_0^{b_t}(v_t - \alpha_0x - (1 - \alpha_0)b_t)g(x)dx \\
    \leq &\int_0^{b_t^*}(v_t - \alpha_0x - (1 - \alpha_0)b_t^*)g(x)dx - \int_0^{b_t^*}(v_t - \alpha_0x - (1 - \alpha_0)b_t^*)\hat{g}_{s-1}(x; \widetilde{\alpha}_{s-1})dx \\
    & + \int_0^{b_t}(v_t - \alpha_0x - (1 - \alpha_0)b_t)\hat{g}_{s-1}(x; \widetilde{\alpha}_{s-1})dx - \int_0^{b_t}(v_t - \alpha_0x - (1 - \alpha_0)b_t)g(x)dx\\
    = & (v_t - b_t^*)\cdot[G(b_t^*) - \hat{G}_{s-1}(b_t^*; \widetilde{\alpha}_{s-1})] - (v_t - b_t)\cdot[G(b_t) - \hat{G}_{s-1}(b_t; \widetilde{\alpha}_{s-1})] + \alpha_0\int_{b_t}^{b_t^*}[G(x) - \hat{G}_{s-1}(x; \widetilde{\alpha}_{s-1})]dx \\
    \leq & |G(b_t^*) - \hat{G}_{s-1}(b_t^*; \widetilde{\alpha}_{s-1})| + |G(b_t) - \hat{G}_{s-1}(b_t; \widetilde{\alpha}_{s-1})| + \sup_x|G(x) - \hat{G}_{s-1}(x; \widetilde{\alpha}_{s-1})|,
\end{align*}
where the first inequality holds by the choice of $b_t$, i.e., 
\begin{align*}
    \int_0^{b_t^*}(v_t - \alpha_0x - (1 - \alpha_0)b_t^*)\hat{g}_{s-1}(x; \widetilde{\alpha}_{s-1})dx \leq \int_0^{b_t}(v_t - \alpha_0x - (1 - \alpha_0)b_t)\hat{g}_{s-1}(x; \widetilde{\alpha}_{s-1})dx.
\end{align*}

Let $\mathrm{Regret}_s$ be the regret achieved in episode $s$. We have
\begin{align*}
    \mathrm{Regret}_s &\leq \sum_{t \in \Gamma_s}|G(b_t^*) - \hat{G}_{s-1}(b_t^*; \widetilde{\alpha}_{s-1})| + \sum_{t \in \Gamma_s}|G(b_t) - \hat{G}_{s-1}(b_t; \widetilde{\alpha}_{s-1})| + T_s\sup_x|G(x) - \hat{G}_{s-1}(x; \widetilde{\alpha}_{s-1})|\\
    & \leq \frac{12AB_2(2+W)T_s}{\underline{\alpha}^2W^2\Tilde{l}_W} + \frac{3B_2\sqrt{T_s}}{B_1} + 3\sqrt{\frac{\log(16S/\delta)T_s}{2}} + 3,
\end{align*}
where the second inequality holds by \Cref{lem: C.3} and \Cref{lem: C.4} (setting $\mathcal{K}_s = \frac{4A}{\underline{\alpha}^2W^2\Tilde{l}_W}$ and $\delta := \delta/2S$). Hence, the inequality holds with probability at least $1 - \delta/S$. Finally, by a union bound over $S$ episodes and the fact that $T_s/T_{s-1} = \sqrt{T}$ and $S \leq \log \log T$, we complete our proof.
\end{proof}

\begin{lemma}
    \label{lem: C.3}
    For each episode $s$, we have with probability at least $1 - \delta / 2S$,
    \begin{align*}
        ||\widetilde{\alpha}_s - \alpha_0|| \leq \frac{4A}{\underline{\alpha}^2W^2\Tilde{l}_W},
    \end{align*}
    where $$A = \frac{W\Tilde{h}_W}{2\overline{\beta}T_s} \cdot \left(\frac{B_2}{B_1\sqrt{T_s}} + \sqrt{\frac{log(16S / \delta)}{2T_s}} + \frac{1}{T_s}\right) + \frac{W\Tilde{h}_W}{2\overline{\beta}}\sqrt{\frac{\log(2S/\delta)}{T_s}}.$$
\end{lemma}

\begin{proof}
    For notation simplicity, we re-parameterize $\alpha$ by denoting $\beta = 1/\alpha$. By \Cref{asm: unknown alpha unknown G}, $\beta$ has an upper bound $\overline{\beta} = 1/\underline{\alpha}$. To compute $\hat{\beta}_s$, we minimize the
    following MLE loss function,
   \begin{align*}
    \calL_s(\beta) = - \frac{1}{T_s} \sum_{t \in \Gamma_s} & \left[ \mathbb{I}\left\{b_t + \frac{W}{2\overline{\beta}} \geq p_t \right\} \log \hat{G}_s\left(b_t + \frac{W}{2\overline{\beta}} \beta; \beta\right) \right. \nonumber \\
    & + \left. \mathbb{I}\left\{b_t + \frac{W}{2\overline{\beta}}  < p_t \right\} \log\left(1 - \hat{G}_s\left(b_t + \frac{W}{2\overline{\beta}} \beta; \beta\right)\right)\right].
    \end{align*}
    We abuse notation and denote $\varepsilon_t(\beta) = b_t + \frac{W}{2\overline{\beta}} \beta$. Observe that
    \begin{align*}
        \bbE \left[\mathbb{I}\left\{b_t + \frac{W}{2\overline{\beta}} \geq p_t \right\}\right] & = \bbE \left[\mathbb{I}\left\{b_t + \frac{W}{2\overline{\beta}} \geq \alpha d_t  + (1-\alpha) b_t \right\}\right] = \bbE \left[\mathbb{I}\left\{b_t + \frac{W}{2\overline{\beta}} \beta \geq d_t\right\}\right] = G(\varepsilon_t(\beta)). 
    \end{align*}
    When $\hat{G}_s$ equals to the real $G$, $\beta_0 = 1/\alpha_0$ should minimize the MLE loss function.
    By the second-order Taylor theorem, we have
    \begin{align*}
       \mathcal{L}_s(\hat{\beta}_s) -  \mathcal{L}_s(\beta_0) = \mathcal{L}' (\beta_0)(\hat{\beta}_s - \beta_0) + \frac{1}{2}\mathcal{L}''(\Tilde{\beta})(\hat{\beta}_s - \beta_0)^2
    \end{align*}
    for some $\Tilde{\beta}$ on the line segment between $\beta$ and $\hat{\beta}_s$. Given the definition of $\mathcal{L}_s(\beta; \cdot)$, we have 
    \begin{align*}
        \mathcal{L}'_s (\beta) = \frac{1}{T_s} \sum_{t \in \Gamma_s}\frac{W}{2\overline{\beta}}\eta_t(\beta), \qquad \mathcal{L}''_s (\beta) = \frac{1}{T_s} \sum_{t \in \Gamma_s}\frac{W^2}{4\overline{\beta}^2}\zeta_t(\beta)
    \end{align*}
    where $\eta_t(\beta)$ and $\zeta_t(\beta)$ are defined as follows,
    \begin{align*}
        \eta_t(\beta) = - \mathbb{I}\left\{b_t + \frac{W}{2\overline{\beta}} \geq p_t \right\} \log' \hat{G}_s(\varepsilon_t(\beta); \beta) - \mathbb{I}\left\{b_t + \frac{W}{2\overline{\beta}} < p_t \right\} \log'(1 - \hat{G}_s(\varepsilon_t(\beta); \beta)),
    \end{align*}
    \begin{align*}
        \zeta_t(\beta) = - \mathbb{I}\left\{b_t + \frac{W}{2\overline{\beta}} \geq p_t \right\} \log'' \hat{G}_s(\varepsilon_t(\beta); \beta) - \mathbb{I}\left\{b_t + \frac{W}{2\overline{\beta}} < p_t \right\} \log''(1 - \hat{G}_s(\varepsilon_t(\beta); \beta)).
    \end{align*}
    Based on our construction of the algorithm, $b_t$ is independent with $d_t$. Therefore, $\varepsilon_t(\beta_0) = b_t + \frac{W}{2\overline{\beta}} \beta_0$ are independent with $d_t$ for any $t \in \Gamma_s$, we have
    \begin{align*}
        \mathbb{E}[\eta_t(\beta_0)] &= -\frac{\hat{g}_s(\varepsilon_t(\beta_0);\beta_0)}{\hat{G}_s(\varepsilon_t(\beta_0);\beta_0)} \cdot G(\varepsilon_t(\beta_0)) + \frac{\hat{g}_s(\varepsilon_t(\beta_0);\beta_0)}{\hat{G}_s(1 - \varepsilon_t(\beta_0);\beta_0)} \cdot (1 - G(\varepsilon_t(\beta_0))) \\
        &= [\hat{G}_s(\varepsilon_t(\beta_0);\beta_0) - G(\varepsilon_t(\beta_0)))] \cdot [\frac{\hat{g}_s(\varepsilon_t(\beta_0);\beta_0)}{\hat{G}_s(\varepsilon_t(\beta_0);\beta_0)} + \frac{\hat{g}_s(\varepsilon_t(\beta_0);\beta_0)}{1 - \hat{G}_s(\varepsilon_t(\beta_0);\beta_0)}].
    \end{align*}
    Thus, by \Cref{lem: C.2} and \Cref{pro: h_W_hat and l_W_hat}, we have
    \begin{align*}
        |\mathbb{E}[\eta_t(\beta_0)]| &=[\hat{G}_s(\varepsilon_t(\beta_0);\beta_0) - G(\varepsilon_t(\beta_0)))] \cdot [\frac{\hat{g}_s(\varepsilon_t(\beta_0);\beta_0)}{\hat{G}_s(\varepsilon_t(\beta_0);\beta_0)} + \frac{\hat{g}_s(\varepsilon_t(\beta_0);\beta_0)}{1 - \hat{G}_s(\varepsilon_t(\beta_0);\beta_0)}] \\
        & \leq 2\Tilde{h}_W \cdot \left(\frac{B_2}{B_1\sqrt{T_s}} + \sqrt{\frac{log(16S / \delta)}{2T_s}} + \frac{1}{T_s}\right)
    \end{align*}
    holds with probability at least $1 - \delta/4S$. Then by Hoeffding’s inequality and union bound
    \begin{align*}
        |\mathcal{L}'_s (\beta_0)| &\leq \frac{W}{2\overline{\beta}T_s}\sum_{t \in \Gamma_s}|\mathbb{E}[\eta_t(\beta_0)]| + \frac{W\Tilde{h}_W}{2\overline{\beta}}\sqrt{\frac{\log(8S/\delta)}{T_s}} \\
        &\leq \frac{W\Tilde{h}_W}{2\overline{\beta}T_s} \cdot \left(\frac{B_2}{B_1\sqrt{T_s}} + \sqrt{\frac{log(16S / \delta)}{2T_s}} + \frac{1}{T_s}\right) + \frac{W\Tilde{h}_W}{2\overline{\beta}}\sqrt{\frac{\log(8S/\delta)}{T_s}} := A
    \end{align*}
    holds with probability at least $1 - \delta/2S$. By the optimality of $\hat{\beta}_s$,
    \begin{align*}
        \mathcal{L}_s(\hat{\beta}_s) -  \mathcal{L}_s(\beta_0) \leq 0.
    \end{align*}
    Invoking into $\mathcal{L}_s(\beta)$, we have
    \begin{align*}
        \frac{1}{2}\mathcal{L}''_s (\beta_0)(\hat{\beta}_s - \beta_0)^2 \leq -\mathcal{L}'_s (\beta_0)(\hat{\beta}_s - \beta_0)
        \leq A|\hat{\beta}_s - \beta_0|,
    \end{align*}
    \begin{align*}
        |\hat{\beta}_s - \beta_0| \leq \frac{2A}{2\Tilde{l}_W \cdot W^2/(4\overline{\beta}^2)} = \frac{4\overline{\beta}^2A}{W^2\Tilde{l}_W}.
    \end{align*}
    holds with probability at least $1 - \delta/2S$.
    So we have
    \begin{align*}
        ||\widetilde{\alpha}_s - \alpha_0|| \leq  \frac{4\overline{\beta}^2A}{W^2\Tilde{l}_W} = \frac{4A}{\underline{\alpha}^2W^2\Tilde{l}_W}.
    \end{align*}
    holds with probability at least $1 - \delta/2S$.
\end{proof}

\begin{lemma}
    \label{lem: C.4}
    For any fixed $\delta > 0$,  suppose $T_s \gg \log^2(2/\delta)$ and conditioned on $|\hat{\beta}_s - \beta_0| \leq \mathcal{K}_s$, we have for all $d \in [0, 1+W]$, 
    \begin{align*}
        |\hat{G}_s(d;\hat{\beta}_s)- G(d)| \leq 3B_2(2+W)\mathcal{K}_s + \frac{B_2}{B_1\sqrt{T_s}} + \sqrt{\frac{\log(8/\delta)}{2T_s}} + \frac{1}{T_s}
    \end{align*}
    holds with probability at least $1 - \delta$.
\end{lemma}

\begin{proof}
    Let $\hat{\mathbb{G}}_s$ be the empirical distribution of samples $\{\hat{\beta}_s(p_t - b_t) + b_t\}_{t \in \Gamma_s}$, i.e.
    \begin{align*}
        \hat{\mathbb{G}}_s(d) = \frac{1}{T_s} \sum_{t \in \Gamma_s}\mathbb{I} \{\hat{\beta}_s(p_t - b_t) + b_t \leq d\} = \frac{1}{T_s} \sum_{t \in \Gamma_s}\mathbb{I}\{d_t \leq d + (\beta_0 - \hat{\beta}_s)(p_t - b_t)\}.
    \end{align*}
    First, we give a uniform convergence bound for $|\hat{\mathbb{G}}_s(d) - G(d)|$. The main challenge is that we cannot directly apply DKW inequality, since $\hat{\beta}_s$ depends on $d_t$, $t \in \Gamma_s$. To handle this challenge, we bound the lower bound and upper bound of $\hat{\mathbb{G}}_s(d)$ separately. Since $|\hat{\beta}_s - \beta_0| \leq \mathcal{K}_s$, and $b_t \leq 1$, $p_t \leq 1+W$, we have
    \begin{align*}
        \frac{1}{T_s} \sum_{t \in \Gamma_s}\mathbb{I}\{d_t \leq d - (2+W)\mathcal{K}_s\} \leq \hat{\mathbb{G}}_s(d) \leq \frac{1}{T_s} \sum_{t \in \Gamma_s}\mathbb{I}\{d_t \leq d + (2+W)\mathcal{K}_s\} 
    \end{align*}
    Thus, conditioned on $|\hat{\beta}_s - \beta_0| \leq \mathcal{K}_s$, for any $\gamma > 0$, we have
    \begin{align*}
        &\mathbb{P}(\hat{\mathbb{G}}_s(d) - G(d+(2+W)\mathcal{K}_s) \leq \gamma) \\
        \geq \; &\mathbb{P}(\frac{1}{T_s} \sum_{t \in \Gamma_s}\mathbb{I}\{d_t \leq d + (2+W)\mathcal{K}_s\} - G(d+(2+W)\mathcal{K}_s) \leq \gamma) \\
        \geq \; & 1 - \mathbb{P}(\sup_d|\frac{1}{T_s} \sum_{t \in \Gamma_s}\mathbb{I}\{d_t \geq d + (2+W)\mathcal{K}_s\} - G(d+(2+W)\mathcal{K}_s)| > \gamma) \\
        \geq \; & 1 - 2exp(-2T_s\gamma^2).
    \end{align*}
    Similarly, we have $\mathbb{P}(\hat{G(d+(2+W)\mathcal{K}_s) - \mathbb{G}}_s(d) \leq \gamma) \geq 1 - 2exp(-2T_s\gamma^2)$ for any $\gamma > 0$, conditioned on $|\hat{\beta}_s - \beta_0| \leq \mathcal{K}_s$. \\
    Therefore, applying a union bound and Lipschitzness of $G$ we have,
    \begin{align}
    \label{equ: first step in Lemma C.4}
        |\hat{\mathbb{G}}_s(d) - G(d)| \leq \sqrt{\frac{\log(8/\delta)}{2T_s}} + B_2(2 + W)\mathcal{K}_s
    \end{align}
    holds with probability at least $1 - \delta/2$. \\
    Second, we apply the similar technique used in \Cref{lem: C.2} to bound $|\hat{G}_s(d;\hat{\beta}_s)- G(d)|$. Denote $\hat{d}_t = \hat{\beta}_s(p_t - b_t) + b_t$, $\forall t \in \Gamma_s$. Thus, for any $\hat{d}_t$, there must exist at least one $d_t$ s.t. $|d_t - \hat{d}_t| = |(\beta_0 - \hat{\beta}_s)(p_t - b_t)|\leq (2+W)\mathcal{K}_s$. Let $r_s = \frac{1}{2B_1\sqrt{T_s}}$, then for any $d \in [0, 1+W]$, the probability that there exists a point $y \in \{\hat{d}_t\}_{t \in \Gamma_s}$ s.t. $|y - z| \leq r_s + (2+W)\mathcal{K}_s$, is at least, 
    \begin{align*}
    1 - (1 - 2B_1 \cdot r_s)^{T_s} = 1 - (1 - \frac{1}{\sqrt{T_s}}) \approx 1 - e^{-\sqrt{T_s}} \geq 1-\delta/2
    \end{align*}
    Therefore, for any $d \in [0, 1 + W]$, we can decompose $\hat{G}_s(d; \hat{\beta}_s) - G(d)$ in the following,
    \begin{align*}
        &|\hat{G}_s(d; \hat{\beta}_s) - G(d)|\\
        \leq &|\hat{G}_s(d; \hat{\beta}_s) - \hat{G}_s(y; \hat{\beta}_s)| + |\hat{G}_s(y; \hat{\beta}_s) - \hat{\mathbb{G}}_s(y)| + |\hat{\mathbb{G}}_s(y)| - G(y)| + |G(y) - G(d)|.
    \end{align*}
    Indeed, the characterization results by \Cref{lem: C.1} applies to samples $\hat{d}_t$. Then we have $|\hat{G}_s(y; \hat{\beta}_s) - \hat{\mathbb{G}}_s(y)| \leq \frac{1}{T_s}$. . By the Lipshitzness of $\hat{G}_s(\cdot; \hat{\beta}_s)$ and $G$, \Cref{equ: first step in Lemma C.4} and union bound, we have
    \begin{align*}
        |\hat{G}_s(d; \hat{\beta}_s) - G(d)| &\leq 2B_2[r_s + (2+W)\mathcal{K}_s] + \frac{1}{T_s} + \sqrt{\frac{\log(8/\delta)}{2T_s}} + B_2(2 + W)\mathcal{K}_s \\
        &= 3B_2(2+W)\mathcal{K}_s + \frac{B_2}{B_1\sqrt{T_s}} + \sqrt{\frac{\log(8/\delta)}{2T_s}} + \frac{1}{T_s}
    \end{align*}
    holds with probability at least $1 - \delta$ when $T_s \gg \log^2(2/\delta)$.
\end{proof}

\end{document}